\documentclass[10pt,oneside]{article}
\usepackage{graphics}
\usepackage{amsmath}
\usepackage{amssymb}
\usepackage{enumerate}
\usepackage{mathrsfs}
\usepackage[dvips]{graphicx}
\usepackage{psfrag}
\newtheorem{theorem}{Theorem}
\newtheorem{lemma}{Lemma}

\newtheorem{prop}{Proposition}
\newenvironment{thrm}{\par{\bf Theorem \rm}\em}{}

\newenvironment{proof}{\par\indent{\emph{Proof.}}}{\hfill$\square$\vspace{1.4mm}\par\noindent}

\newcommand{\X}{\mathbf{x}}
\newcommand{\E}{\boldsymbol{\eta}}
\newcommand{\Y}{\mathbf{y}}
\newcommand{\XI}{\boldsymbol{\xi}}

\newcommand{\Oj}{\mathbf{O}}
\newcommand{\bZ}{\mathbf{z}}

\newcommand{\cN}{\mathcal{N}}
\newcommand{\cD}{\mathcal{D}}

\newcommand{\N}{\boldsymbol{n}}

\newcommand{\bcN}{\mathbf{N}}

\newcommand{\boxedeqn}[1]{%
\[\fbox{%
\addtolength{\linewidth}{-2\fboxsep}%
\addtolength{\linewidth}{-2\fboxrule}%
\begin{minipage}{\linewidth}%
\begin{equation}#1\end{equation}%
\end{minipage}%
}\]%
}

\numberwithin{equation}{section}

\title{Green's kernels for transmission problems  in bodies with small inclusions}
\author{V. Maz'ya\footnote{Department of Mathematical Sciences, University of Liverpool, Liverpool L69 3BX, U.K., and Department of Mathematics, Link\"oping University, SE-581 83 Link\"oping, Sweden.}, A. Movchan\footnote{Department of Mathematical Sciences, University of Liverpool, Liverpool L69 3BX, U.K.}, M. Nieves$^\dagger$\\
\\ \\
\emph{In memory of V. B. Lidskii}}
\date{}

\begin{document}

\maketitle
\begin{abstract}
The uniform asymptotic approximation of Green's kernel for the transmission problem of antiplane shear  is obtained for domains with small inclusions. 
The remainder estimates are provided. Numerical simulations are  presented to illustrate the effectiveness of the approach.
\end{abstract}

\section{Introduction}\label{transintro}
Our goal is to obtain a uniform asymptotic approximation of Green's function for a transmission problem  
of antiplane shear 
 in  a domain with small inclusions.

Exact solutions to singularly perturbed problems corresponding to bodies with defects are often unavailable. For complicated geometries
involved in problems of this kind, i.e. domains with multiple small perforations, numerical algorithms may become incapable of reaching the required accuracy. Also, when the right-hand sides 
of such boundary value problems  are singular, numerical procedures can  suffer from the same deficiencies. In this case, asymptotic 
solutions to these problems are desirable.

The approximation of Green's kernels for regularly perturbed problems, for 
the Laplace operator and the biharmonic operator, was first studied by Hadamard in \cite{Hadv1}.
More recently, the question of uniformity for the approximations of Green's kernels for boundary value problems in domains with singularly and regularly perturbed boundaries, was addressed in \cite{8}.
The uniform approximations in \cite{8} were derived using the method of compound asymptotic expansions.

The paper \cite{1} contains the rigorous proofs and remainder estimates for uniform approximations of Green's kernels, given in \cite{8},  for $-\Delta$ in an $n$-dimensional domain $(n\ge 2)$, with a single small rigid inclusion. 
Uniform asymptotic formulae of Green's functions for  mixed problems of antiplane shear in domains with a small hole or a crack, are discussed in detail in \cite{MazMovmix}. 
In \cite{MMN1}, an extension of the theory developed in \cite{1} is made to the case of Green's tensors of vector elasticity, for an elastic body with a small inclusion. This was followed by \cite{MMN2}, where uniform asymptotics of Green's kernels for planar and three-dimensional elasticity in bodies with multiple rigid inclusions are given.
The paper \cite{MMN2} also includes analysis of Green's kernels and  numerical simulations for   antiplane shear and plane strain.
  
In the present paper, the new feature of the problem tackled 
is that on the small inclusions we prescribe transmission conditions (the continuity of tractions and displacements). The inclusions are assumed to be occupied by materials which are different from that of the ambient medium. Compared to previous expositions into the uniform approximation of Green's kernels in \cite{1, MazMovmix, MMN1}, where the kernels are approximated in the bodies containing small holes 
we also must approximate the Green's kernel inside the inclusions. The analysis also brings additional boundary layers when the point force is placed inside the inclusion.

 Below, we illustrate one of the main results in this article, for the case when the domain has a single inclusion.
Let $\omega_\varepsilon$ be a small planar inclusion, occupied by a material of shear modulus $\mu_I$, containing the origin $\Oj$ and  which is perfectly bonded to the rest of the matrix $\Omega_\varepsilon \subset \mathbb{R}^2$ whose shear modulus is $\mu_O$. Here, $\varepsilon$ is a small positive parameter characterising the normalized size of the inclusion. 
Consider the antiplane shear Green's function $N_\varepsilon$ for the transmission problem 
 inside the perturbed domain $\Omega_\varepsilon \cup \omega_\varepsilon$. We also use $N^{(\Omega)}$, as the Neumann function in the unperturbed domain (without the inclusion) and $R^{(\Omega)}$ as its regular part. We denote by $\cN$  the Green's function for the transmission problem,  in the unbounded domain with  the scaled inclusion containing the origin. Also let us define the  vector function  $\cD=(\cD_1, \cD_2)^T$, where the components $\cD_j$, $j=1,2$ are the dipole fields for the scaled inclusion in the unbounded domain. 
 
As one of the  results we  state

\vspace{0.1in}\begin{thrm}
The  approximation of Green's function  for the 
transmission problem  of antiplane shear in $\Omega_\varepsilon\cup \omega_\varepsilon\subset \mathbb{R}^2$, 
is given by
\boxedeqn{\label{introtransmr}\begin{array}{c}
\displaystyle{N_\varepsilon(\X, \Y)=N^{(\Omega)}(\X, \Y)+ \cN(\varepsilon^{-1}\X, \varepsilon^{-1}\Y)+(2\pi\mu_O)^{-1} \log(\varepsilon^{-1}|\X-\Y|)}\\ \\
\displaystyle{+\varepsilon \cD(\varepsilon^{-1}\X) \cdot \nabla_\X R^{(\Omega)}(\Oj, \Y)+ \varepsilon\cD(\varepsilon^{-1}\Y) \cdot \nabla_\Y R^{(\Omega)}(\X, \Oj)+O(\varepsilon^2)}
\end{array}}
uniformly for $\X, \Y\in \Omega_\varepsilon \cup \omega_\varepsilon$.
\end{thrm}

\vspace{0.1in}The structure of the article is as follows. 
In Section \ref{transnotmult}, we introduce the main notations that will be adopted throughout the text and define  Green's function for  the domain containing several inclusions. Section \ref{modfield} contains the description of  model fields used to construct the uniform asymptotic approximations of the Green's function. In Section \ref{transunbest}, we  state and prove an estimate for solutions to model transmission problems in an unbounded domain with an inclusion. Solutions to transmission problems in a domain with several small inclusions are studied in   Section \ref{transestmultiplinc}.
 Results of these sections will aid us in deriving the remainder estimate present in the generalization of  approximation (\ref{introtransmr}) to the case of multiple inclusions.
Asymptotic properties of the boundary layer fields, involving the regular part of Green's function for the transmission problem,
in the unbounded domain with an inclusion
 are investigated in Sections \ref{transregest}. 
Then,  we  consider the uniform approximation of  Green's function for the transmission problem  in the domain with  small inclusions in Section \ref{transmultot}, and give the formal algorithm together with the remainder estimates.
 Finally, in  Section \ref{antiplnumsim},
we demonstrate the effectiveness of our approach and present the numerical simulations comparing the asymptotic formula in Section \ref{transmultot} with the finite element calculations in COMSOL.

The asymptotic formulae obtained in the sequel are readily applicable to numerical simulations. As an example, Figure  \ref{figtrans} shows the comparison between an asymptotic approximation  and COMSOL computation for the modulus of the gradient of the  regular part of Green's function for the
transmission problem in a domain with a circular inclusion, for the case when the point force is applied  outside 
 the inclusion in a planar body. This Figure represents the regular part of the displacement field produced by a point force inside a Cast Iron disk containing an Aluminum inclusion:   Fig. \ref{figtrans}a, shows  the computations obtained through the formula  (\ref{introtransmr}) when $\Y \in \Omega_\varepsilon$, while Fig. \ref{figtrans}b corresponds to the numerical finite element solution produced  in COMSOL. The surface plots shown are very similar.

\begin{figure}[htbp]
\begin{minipage}[b]{0.5\linewidth}
\centering
\includegraphics[width=\textwidth]{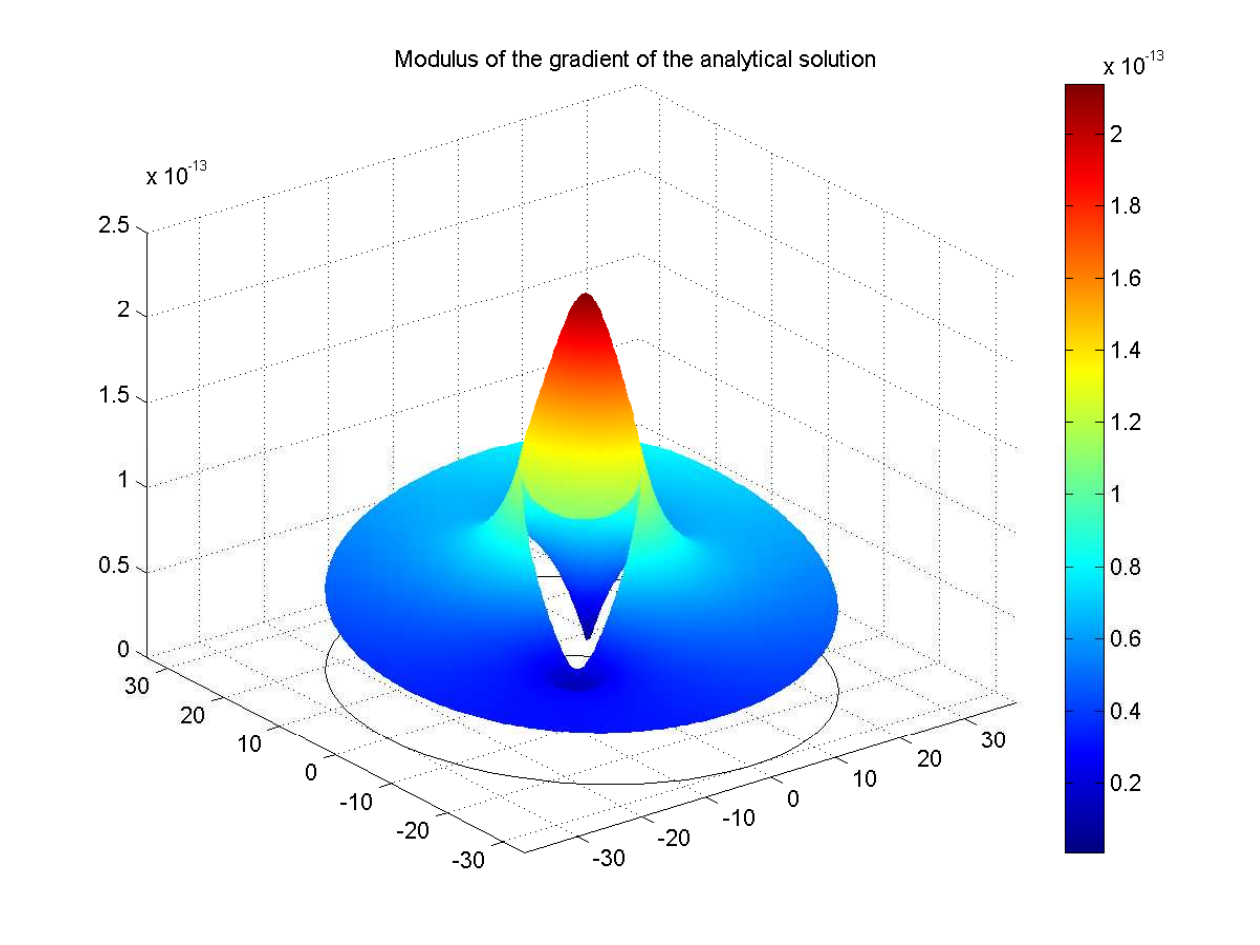}
a)
\end{minipage}
\begin{minipage}[b]{0.5\linewidth}
\centering
\includegraphics[width=\textwidth]{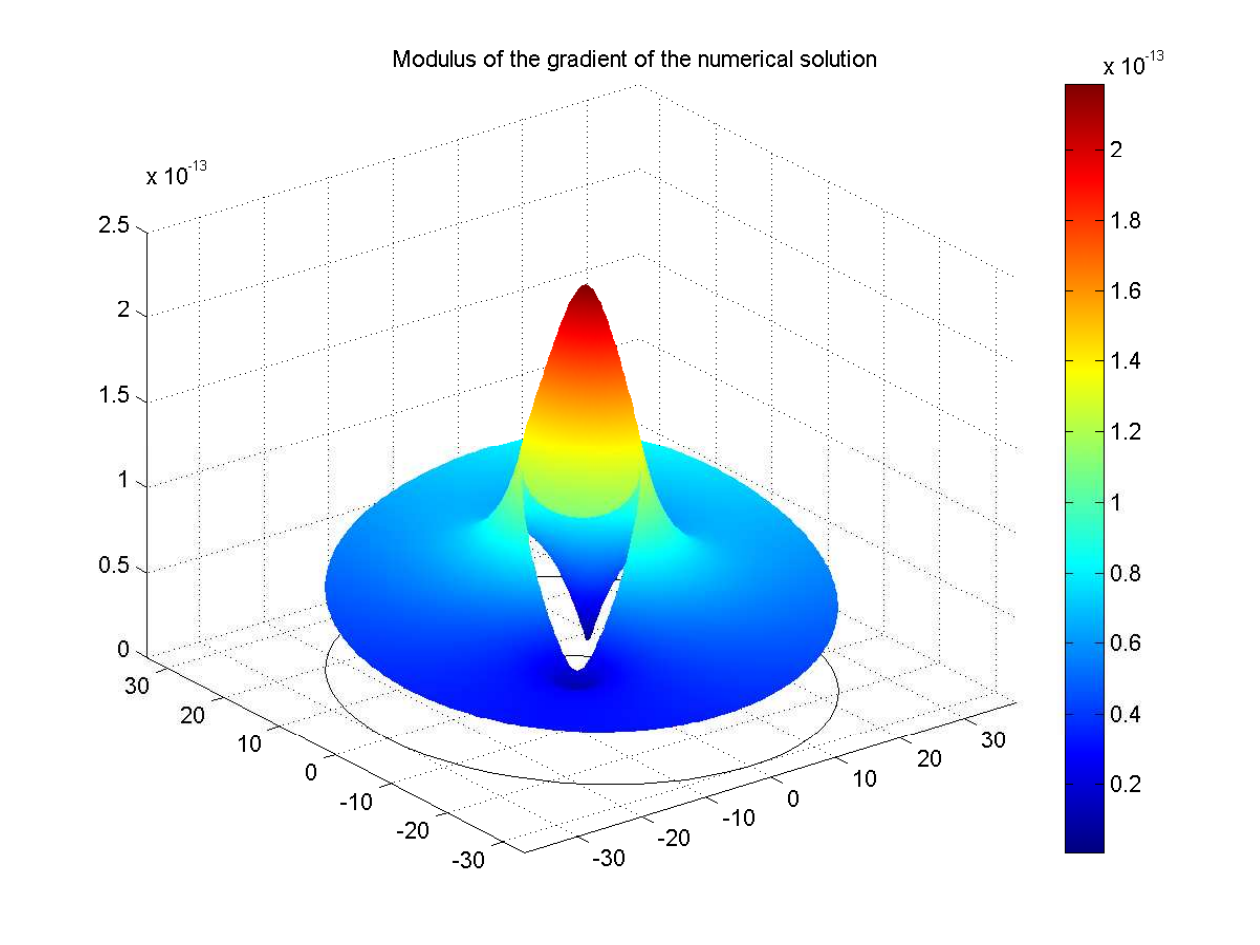}
b)
\end{minipage}
\caption{a) Modulus of the gradient of the regular part of the Green's function, computed with  the aid of the asymptotic formula (\ref{introtransmr}) when $\Y \in \Omega_\varepsilon$, b) A finite element computation (produced in COMSOL) for the modulus of the gradient of the regular part of the Green's function for the transmission problem in a domain with an inclusion. Here, we consider a circular cylinder of radius 30m containing a circular inclusion of radius 7m, the shear modulus of the inclusion is $\mu_I=2.6316 \times 10^{10} \text{Nm}^{-2}$ (an Aluminum inclusion), where the shear modulus for the rest of the matrix is $\mu_O=5.6\times 10^{10} \text{Nm}^{-2}$. The position of the unit point force is $\Y =(10\text{m}, 10\text{m})$, which is quite close to the boundary of the inclusion.}
\label{figtrans}
\end{figure}


\section{Main Notations}\label{transnotmult}
Let $\Omega$ be a subset of $\mathbb{R}^2$, with smooth boundary $\partial \Omega$ and compact closure $\bar{\Omega}$. Also let   $\omega^{(j)} \subset \mathbb{R}^2$, have  smooth boundary $\partial \omega^{(j)}$ and compact closure $\bar{\omega}^{(j)}$, whose complement  in the infinite plane is  $C \bar{\omega}^{(j)}=\mathbb{R}^2 \backslash \bar{\omega}^{(j)}$, $j=1, \dots, M$. The sets  $\omega^{(j)}$, $j=1, \dots, M$, are assumed to contain the origin $\Oj$, and the maximum distance between $\Oj$ and $\partial \omega^{(j)}$ is 1. Let $\omega^{(j)}_\varepsilon$, be a subset of $\Omega$, with centre $\Oj^{(j)}$,  $1\le j\le M$. We assume that the minimum distance between $\Oj^{(j)}$ and $\Oj^{(k)}$, $k \ne j$, $k=1, \dots, M$ and the minimum distance between $\Oj^{(j)}$  and  $\partial \Omega$ is 1.
We relate the domain $\omega_\varepsilon^{(j)}$ to $\omega^{(j)}$ via $\omega_\varepsilon^{(j)}=\{ \X: \varepsilon^{-1} (\X-\Oj^{(j)}) \in \omega^{(j)}\}$, $j=1, \dots, M$. The perturbed geometry is defined by $\Omega_\varepsilon=\Omega \backslash \overline{\bigcup_j \omega_\varepsilon^{(j)}}$, and we say that this domain is occupied by a material with shear modulus $\mu_O$ and the domain $\omega^{(j)}_\varepsilon$ is occupied by a material with shear modulus $\mu_{I_j}$, where $\mu_O, \mu_{I_j}>0$, $1 \le j\le  M$. In the subsequent  sections, along with $\X$ and $\Y$ we will also use the scaled variables $\XI_j=\varepsilon^{-1} (\X-\Oj^{(j)})$, $\E_j=\varepsilon^{-1}(\Y-\Oj^{(j)})$.

By $\chi_T$ we mean the characteristic function of the set $T$, that is
\begin{equation*}\label{charfn}
\chi_T(\X)= \left\{ \begin{array}{ll}
1\;, \quad \text{ if } \X \in T\;,\\
0\;, \quad \text{ otherwise}\;.\end{array}\right.
\end{equation*}
Our goal is to obtain uniform asymptotics for the Green's function $N_\varepsilon$, of the transmission problem in $ \bigcup_j \omega_\varepsilon^{(j)}\cup \Omega_\varepsilon$, which is a solution of
\begin{equation}\label{transnotmultintro1}
\mu_O \Delta_\X N_\varepsilon(\X, \Y)+\delta(\X-\Y)=0\;, \quad \X \in \Omega_\varepsilon\;, \Y \in  \bigcup_l \omega_\varepsilon^{(l)} \cup \Omega_\varepsilon\;, 
\end{equation} 
\begin{equation}\label{transnotmultintro2}
\mu_{I_j} \Delta_\X N_\varepsilon(\X, \Y)+\delta(\X-\Y)=0\;, \quad \X \in \omega_\varepsilon^{(j)}\;, j=1, \dots, M, \Y \in \bigcup_l \omega_\varepsilon^{(l)} \cup \Omega_\varepsilon\;.
\end{equation} 
The normal derivative of $N_\varepsilon$ on the exterior boundary satisfies 
\begin{equation}\label{transnotmultintro3}
 \mu_O \frac{\partial N_\varepsilon}{\partial n_\X}(\X, \Y)=-\frac{1}{|\partial \Omega|}\;, \quad  \X \in \partial \Omega\;, \Y \in \bigcup_l \omega_\varepsilon^{(l)} \cup \Omega_\varepsilon\;,
\end{equation}
where  $|\partial \Omega|$ is the one-dimensional measure of the set $\partial \Omega$; $\partial/\partial n_\X=\N \cdot \nabla_\X$ is the normal derivative, where $\N$ is the unit outward  normal.

Assuming  that the small inclusions $\omega^{(j)}_\varepsilon$, $j=1, \dots, M$ are  perfectly bonded to the matrix $\Omega_\varepsilon$, we write  transmission conditions across the interface $\partial \omega_\varepsilon^{(j)}$ in the form
\begin{equation}\label{transnotmultintro4} \mu_O  \frac{\partial N_\varepsilon}{\partial n_{\X}}(\X, \Y)\big|_{\X \in \partial \omega^{(j)+}_\varepsilon}=\mu_{I_j}  \frac{\partial N_\varepsilon}{\partial n_{\X}}(\X, \Y)\big|_{\X \in \partial \omega^{(j)-}_\varepsilon}\;, 
\end{equation}
\begin{equation*}\label{transnotmultintro5}  
N_\varepsilon(\X, \Y)\big|_{\X \in \partial \omega^{(j)+}_\varepsilon}= N_\varepsilon(\X, \Y)\big|_{\X \in \partial \omega^{(j)-}_\varepsilon}\;,
\end{equation*}
for $j=1, \dots, M$ and  $\Y \in \bigcup_l \omega_\varepsilon^{(l)} \cup \Omega_\varepsilon$; the notation $\partial \omega^{(j)\pm}_\varepsilon$  indicates the exterior or interior boundary of the set $\omega^{(j)}_\varepsilon$.

The symmetry of $N_\varepsilon$, i.e. 
\[ N_\varepsilon(\X, \Y)=N_\varepsilon(\Y, \X)\;,\]
is  guaranteed by the condition
\begin{equation}\label{transnotmultintro7}
\int_{\partial \Omega} N_\varepsilon(\X, \Y)\, dS_\X=0\;.
\end{equation}
The proof of symmetry of  $N_\varepsilon$, for this problem, is addressed below.

\vspace{0.1in}\noindent\bf The symmetry of the Green's function $N_\varepsilon$ \rm
  
 \begin{prop}\label{propsym} Green's function $N_\varepsilon$ for the transmission problem in $\bigcup_{l=1}^M \omega^{(l)}_\varepsilon \cup \Omega_\varepsilon$  satisfies
 \begin{equation}\label{symprop}
 N_\varepsilon(\X, \Y)=N_\varepsilon(\Y, \X)\;.
 \end{equation} 
 \end{prop}
 
 \begin{proof} Let
 \begin{equation}\label{2.6a}
 \begin{array}{ll} N^{(O)}_\varepsilon(\X, \Y)=N_\varepsilon(\X, \Y)\;, &\quad \text{ for } \X \in \Omega_\varepsilon, \Y \in \bigcup_{l=1}^M \omega^{(l)}_\varepsilon\cup\Omega_\varepsilon\;,\\
 N^{(I_j)}_\varepsilon(\X, \Y) =N_\varepsilon(\X, \Y)\;,& \quad \text{ for }\X \in \omega^{(j)}_\varepsilon, \Y \in \bigcup_{l=1}^M \omega^{(l)}_\varepsilon\cup\Omega_\varepsilon\;.
 \end{array}
 \end{equation}
 We apply Green's formula for the functions $N^{(O)}_\varepsilon$ and $N^{(I_j)}_\varepsilon$, $j=1, \dots, M$:
 \begin{eqnarray*}
 &&\mu_O \int_{\Omega_\varepsilon}\Big\{ N^{(O)}_\varepsilon(\bZ, \X)\Delta_{\bZ} N^{(O)}_\varepsilon(\bZ, \Y)-N^{(O)}_\varepsilon(\bZ, \Y) \Delta_{\bZ}N^{(O)}_\varepsilon(\bZ, \X)\Big\}\, d\bZ\\
 &&+\sum_{l=1}^M\mu_{I_l} \int_{\omega_\varepsilon^{(l)}}\Big\{ N^{(I_l)}_\varepsilon(\bZ, \X)\Delta_{\bZ} N^{(I_l)}_\varepsilon(\bZ, \Y)-N^{(I_l)}_\varepsilon(\bZ, \Y) \Delta_{\bZ}N^{(I_l)}_\varepsilon(\bZ, \X)\Big\}\, d\bZ\\
 &=&\mu_O \int_{\partial\Omega}\Big\{ N^{(O)}_\varepsilon(\bZ, \X)\frac{\partial  N^{(O)}_\varepsilon}{\partial n_{\bZ}}(\bZ, \Y)-N^{(O)}_\varepsilon(\bZ, \Y) \frac{\partial  N^{(O)}_\varepsilon}{\partial n_{\bZ}}(\bZ, \X)\Big\}\, dS_{\bZ}\\ 
&&+\sum_{l=1}^M\int_{\partial \omega_\varepsilon^{(l)}}\left\{ N^{(O)}_\varepsilon(\bZ, \X) \left[\mu_{I_l} \frac{\partial N^{(I_l)}_\varepsilon}{\partial n_{\bZ}}(\bZ, \Y)-\mu_O \frac{\partial N^{(O)}_\varepsilon}{\partial n_{\bZ}}(\bZ, \Y)\right] \right.\\
&&\left.-N^{(I_l)}_\varepsilon(\bZ, \Y) \left[\mu_{I_l} \frac{\partial N^{(I_l)}_\varepsilon}{\partial n_{\bZ}}(\bZ, \X)-\mu_O \frac{\partial N^{(O)}_\varepsilon}{\partial n_{\bZ}}(\bZ, \X)\right]\right\}dS_{\bZ}\;,
 \end{eqnarray*}
 where  the right-hand side is zero as a result of the transmission conditions (\ref{transnotmultintro4}), the exterior boundary condition (\ref{transnotmultintro3}) and the normalization (\ref{transnotmultintro7}). Thus
\begin{eqnarray}
 0&=&\mu_O \int_{\Omega_\varepsilon}\Big\{ N^{(O)}_\varepsilon(\bZ, \X)\Delta_{\bZ} N^{(O)}_\varepsilon(\bZ, \Y)-N^{(O)}_\varepsilon(\bZ, \Y) \Delta_{\bZ}N^{(O)}_\varepsilon(\bZ, \X)\Big\}\, d\bZ\nonumber\\
 &&+\sum_{l=1}^M\mu_{I_l} \int_{\omega_\varepsilon^{(l)}}\Big\{ N^{(I_l)}_\varepsilon(\bZ, \X)\Delta_{\bZ} N^{(I_l)}_\varepsilon(\bZ, \Y)-N^{(I_l)}_\varepsilon(\bZ, \Y) \Delta_{\bZ}N^{(I_l)}_\varepsilon(\bZ, \X)\Big\}\, d\bZ\;.\label{sym1}
\end{eqnarray}
The next step involves using the governing equations (\ref{transnotmultintro1}) and (\ref{transnotmultintro2}) along with the above definitions of $N^{(O)}_\varepsilon$ and $N^{(I_j)}_\varepsilon$, $j=1, \dots, M$.
When $\X ,\Y \in \Omega_\varepsilon$, (\ref{sym1}) gives
\[ N^{(O)}_\varepsilon(\X, \Y)=N^{(O)}_\varepsilon(\Y, \X)\]
whereas if $\X \in \Omega_\varepsilon$, $\Y \in \omega^{(j)}_\varepsilon$
\[ N^{(O)}_\varepsilon(\X, \Y)=N^{(I_j)}_\varepsilon(\Y, \X)\;,\quad  j=1,\dots, M\;.\]
Similarly, for $\X \in \omega^{(j)}_\varepsilon$, $\Y \in \Omega_\varepsilon$, we deduce
\[ N^{(I_j)}_\varepsilon(\X, \Y)=N^{(O)}_\varepsilon(\Y, \X)\;,\quad j=1, \dots, M\;,\]
and if $\X, \Y \in \omega_\varepsilon^{(j)}$
 \[ N^{(I_j)}_\varepsilon(\X, \Y)=N^{(I_j)}_\varepsilon(\Y, \X)\;, \quad j=1, \dots, M\;. \]
 Finally, with $\X \in \omega^{(j)}_\varepsilon$, $\Y \in \omega_\varepsilon^{(k)}$, $k \ne j$, we have
 \[N_\varepsilon^{(I_j)}(\X, \Y)=N_\varepsilon^{(I_k)}(\Y, \X),\quad 1 \le j, k\le M,\quad k \ne j\;.\]
 The above relations together with (\ref{2.6a}) lead to (\ref{symprop}).
\end{proof}

\section{Special solutions  in model domains}\label{modfield}
The asymptotic algorithm uses special fields defined in model domains including the unperturbed set and the 
exterior of a scaled inclusion.

\begin{enumerate}[1.]
\vspace{0.1in}\item The regular part $R^{(\Omega)}$ of the Neumann  function $N^{(\Omega)}$  in  $\Omega$ is defined as a solution of 
\begin{equation*}\label{trans1}
 \mu_O \Delta_\X R^{(\Omega)}(\X, \Y)=0\;,\quad \X\in \Omega\;,
\end{equation*}
\begin{equation}\label{trans2}
 \mu_O \frac{\partial R^{(\Omega)}}{\partial n_\X}(\X, \Y)=-\frac{\partial }{\partial n_\X}((2\pi)^{-1} \log |\X-\Y|)+\frac{1}{|\partial \Omega|}\;, \quad  \X \in \partial \Omega\;,
\end{equation}
where $\Y \in \Omega$.

To guarantee the symmetry of $R^{(\Omega)}$ we impose the orthogonality condition
\[\int_{\partial \Omega} R^{(\Omega)}(\X, \Y)\,d S_\X=-\frac{1}{2\pi\mu_O} \int_{\partial \Omega}\log |\X-\Y|\,d S_\X\;.\]
The function $N^{(\Omega)}$ is related to $R^{(\Omega)}$ by
\begin{equation}\label{defN}
N^{(\Omega)}(\X, \Y)=-(2\pi\mu_O)^{-1}\log|\X-\Y|-R^{(\Omega)}(\X, \Y)\;.
\end{equation}

\vspace{0.1in}\item The next model field is Green's function  for the transmission problem in the domain $C\bar{\omega}^{(j)} \cup \omega^{(j)}$, $j=1, \dots, M$. This function is denoted by $\cN^{(j)}$ and, for $\E\in C\bar{\omega}^{(j)} \cup \omega^{(j)}$, is subject to 
\begin{equation*}\label{trans3}
\mu_O \Delta_{\XI} \cN^{(j)}(\XI, \E)+\delta(\XI-\E)=0\;,\quad \XI \in C\bar{\omega}^{(j)}\;,
\end{equation*}
\begin{equation*}\label{trans4}
 \mu_{I_j} \Delta_{\XI} \cN^{(j)} (\XI, \E)+\delta(\XI-\E)=0\;, \quad \XI \in \omega^{(j)}\;,
\end{equation*}
where the  transmission conditions  across the interface of the inclusion are given as
\begin{equation}\label{trans5} 
\begin{array}{c}
\displaystyle{\mu_O  \frac{\partial \cN^{(j)}}{\partial n_{\XI}}(\XI, \E)\big|_{\XI \in \partial \omega^{(j)+}}=\mu_{I_j}  \frac{\partial \cN^{(j)}}{\partial n_{\XI}}(\XI, \E)\big|_{\XI \in \partial \omega^{(j)-}}\;,} \\
\\
\displaystyle{\cN^{(j)}(\XI, \E)\big|_{\XI \in \partial \omega^{(j)+}}= \cN^{(j)}(\XI, \E)\big|_{\XI \in \partial \omega^{(j)-}}}\;,
\end{array}
\end{equation}
and at infinity we will prescribe the condition
\begin{equation}\label{trans7a}
\cN^{(j)}(\XI, \E)= -(2\pi \mu_O)^{-1}\log|\XI|+c(\E)+O(|\XI|^{-1}), \quad \text{ as }\quad|\XI|\to \infty\;.
\end{equation}
\emph{Symmetry of $\cN^{(j)}$.} We choose
\begin{equation}\label{symcN1}
c(\E)\equiv 0\;,
\end{equation}
and then it can be shown that 
\[ \cN^{(j)}(\XI, \E)=\cN^{(j)}(\E, \XI)\;,\]
i.e. $\cN^{(j)}$ is symmetric.

The proof is analogous to the one of Proposition \ref{propsym}.

We set 
\[\cN^{(j,O)}(\XI, \E)=\cN^{(j)}(\XI, \E)\quad \text{for }\E\in C\bar{\omega}^{(j)}, \XI \in C\bar{\omega}^{(j)} \cup \omega^{(j)}\;,\]
\[\cN^{(j,I)}(\XI, \E)=\cN^{(j)}(\XI, \E)\quad \text{for }\E\in \omega^{(j)}, \XI \in C\bar{\omega}^{(j)} \cup \omega^{(j)}\;,\]
and
\[ c^{(O)}(\E)=c(\E) \text{ for } \E \in C\bar{\omega}^{(j)}\;,\]
\[c^{(I)}(\E)=c(\E) \text{ otherwise }.\]
Assume $\E\in C\bar{\omega}^{(j)}$, and let $B_R=\{\XI:|\XI|<R\}$ be a disk, with sufficiently large radius $R$, so that $\E \in B_R \backslash \bar{\omega}^{(j)}$. By applying Green's formula to $\cN^{(j)}(\bZ, \XI)$ and $\cN^{(j)}(\bZ, \E)$ in $B_R \backslash \bar{\omega}^{(j)}$ and $\omega^{(j)}$ we obtain
\begin{eqnarray*}
&&\mu_O \int_{B_R\backslash \bar{\omega}^{(j)}}\Big\{\cN^{(j)}(\bZ, \XI)\Delta_\bZ \cN^{(j,O)}(\bZ,\E)-\cN^{(j, O)}(\bZ, \E)\Delta_\bZ \cN^{(j)}(\bZ, \XI)\Big\}\,d\bZ\\
&&+\mu_{I_j}\int_{\omega^{(j)}}\Big\{\cN^{(j)}(\bZ, \XI)\Delta_\bZ \cN^{(j,O)}(\bZ,\E)-\cN^{(j)}(\bZ, \E)\Delta_\bZ \cN^{(j)}(\bZ, \XI)\Big\}\,d\bZ\\
&&=\mu_O \int_{\partial B_R}\Big\{\cN^{(j)}(\bZ, \XI)\frac{\partial\cN^{(j,O)}}{\partial n_{\bZ}}(\bZ,\E)-\cN^{(j, O)}(\bZ, \E)\frac{\partial \cN^{(j)}}{\partial n_{\bZ}}(\bZ, \XI)\Big\}\,dS_{\bZ}\;.
\end{eqnarray*}
The transmission conditions (\ref{trans5}) for $\cN^{(j)}$ imply that the integral over $\partial \omega^{(j)}$ is zero. When $\XI \in B_R \backslash \bar{\omega}^{(j)}$,
\[
\begin{array}{c}
\cN^{(j)}(\XI, \E)-\cN^{(j)}(\E, \XI)\\ \\
\displaystyle{= \mu_O\int_{\partial B_R}\Big\{\cN^{(j,O)}(\bZ, \XI)\frac{\partial\cN^{(j,O)}}{\partial n_{\bZ}}(\bZ,\E)-\cN^{(j, O)}(\bZ, \E)\frac{\partial \cN^{(j,O)}}{\partial n_{\bZ}}(\bZ, \XI)\Big\}\,dS_{\bZ}\;.}\end{array}\]
Taking the limit as $R\to \infty$ and employing (\ref{trans7a}), we deduce 
\[\begin{array}{c}
\displaystyle{\cN^{(j)}(\XI, \E)-\cN^{(j)}(\E, \XI)}\\ \\
\displaystyle{=\lim_{R\to\infty}\mu_O \int_{\partial B_R}\Big\{ \Big( \frac{1}{2\pi\mu_O}\log|\bZ|^{-1}+c^{(O)}(\XI)+O\left(\frac{1}{R}\right)\Big) }\\\\
\displaystyle{\times\Big( \frac{1}{2\pi\mu_O}\frac{\partial \log |\bZ|^{-1}}{\partial n_{\bZ}}+O\left(\frac{1}{R^2}\right)\Big)}\\ \\
\displaystyle{-\Big( \frac{1}{2\pi\mu_O}\log|\bZ|^{-1}+c^{(O)}(\E)+O\left(\frac{1}{R}\right)\Big) \Big( \frac{1}{2\pi\mu_O}\frac{\partial \log |\bZ|^{-1}}{\partial n_{\bZ}}+O\left(\frac{1}{R^2}\right)\Big)\Big\}dS_{\bZ}\;,}
\end{array}\]
which is equivalent to 
\begin{eqnarray*}
\cN^{(j)}(\XI, \E)-\cN^{(j)}(\E, \XI)&=&(c^{(O)}(\E)-c^{(O)}(\XI))\lim_{R\to\infty}\int_{\partial B_R}\frac{\partial }{\partial n_\bZ}((2\pi)^{-1}\log|\bZ|)dS_{\bZ}\\
&=& c^{(O)}(\E)-c^{(O)}(\XI)=0\;.
\end{eqnarray*}
Hence $\cN^{(j)}(\XI, \E)$ is symmetric for $\XI, \E \in C\bar{\omega}^{(j)}$. In a similar way, it can be shown that
\[\cN^{(j)}(\XI, \E)-\cN^{(j)}(\E, \XI)=c^{(I)}(\E)-c^{(O)}(\XI) \text{ for }\XI \in C \bar{\omega}^{(j)}, \E \in \omega^{(j)}\;,\]
\[\cN^{(j)}(\XI, \E)-\cN^{(j)}(\E, \XI)=c^{(O)}(\E)-c^{(I)}(\XI) \text{ for }\E \in C \bar{\omega}^{(j)}, \XI \in \omega^{(j)}\;,\]
\[\cN^{(j)}(\XI, \E)-\cN^{(j)}(\E, \XI)=c^{(I)}(\E)-c^{(I)}(\XI) \text{ for }\XI \in \omega^{(j)}, \E \in \omega^{(j)}\;,\]
and the condition (\ref{symcN1}) implies that the above right-hand sides are zero. Thus $\cN^{(j)}$ is symmetric.

\emph{Regular part of $\cN^{(j)}$. }Let the function $\cN^{(j)}$ have the form
\begin{equation*}
\cN^{(j)}(\XI, \E)=\chi_{C\bar{\omega}^{(j)}}(\E)\cN^{(j,O)}(\XI, \E)+\chi_{\omega^{(j)}}(\E)\cN^{(j,I)}(\XI, \E)\;,
\end{equation*}
where 
\begin{equation}\label{trans7}
\begin{array}{c}
\displaystyle{\cN^{(j,O)}(\XI, \E)=-(2\pi\mu_O)^{-1}\log|\XI-\E|-h_N^{(j,O)}(\XI, \E)}\;,\\
\displaystyle{\cN^{(j,I)}(\XI, \E)=-(2\pi\mu_{I_j})^{-1}\log|\XI-\E|-h_N^{(j,I)}(\XI, \E)}\;,
\end{array}
\end{equation}
and here $h_N^{(j,O)}$ and $h_N^{(j, I)}$ are the regular parts of $\cN^{(j, O)}$ and $\cN^{(j, I)}$, $j=1, \dots, M$, respectively. Moreover, we also set
\begin{eqnarray}\label{eqhnIOj}
&&\begin{array}{ll}h^{(j, O, O)}_N(\XI, \E)=h^{(j, O)}_N(\XI, \E)  \;, &\quad \text{ for } \XI, \E\in C\bar{\omega}^{(j)}\;,\\
 h_N^{(j, I, O)}(\XI, \E)=h_N^{(j,O)}(\XI, \E)\;, & \quad \text{ for } \XI \in \omega^{(j)}, \E \in C\bar{\omega}^{(j)}\;, \text{ and}
 \end{array}\nonumber\\ \\
 &&\begin{array}{ll}
  h_N^{(j, O, I)}(\XI, \E)=h_N^{(j, I)}(\XI, \E) \;, &\quad \text{ for } \XI\in C\bar{\omega}^{(j)}, \E \in \omega^{(j)}\;,\\
 h_N^{(j,I, I)}(\XI, \E)=h_N^{(j, I)}(\XI, \E)\;, & \quad \text{ for } \XI , \E \in \omega^{(j)}\;.
 \end{array}\nonumber
 \end{eqnarray}
Then, the above definitions for $h_N^{(j, O)}$ and $h_N^{(j, I)}$ lead to this representation for 
$\cN^{(j)}$
\begin{eqnarray*}
\cN^{(j)}(\XI, \E)&=&\chi_{C\bar{\omega}^{(j)}}(\XI) \chi_{C\bar{\omega}^{(j)}}(\E) \left\{-\frac{1}{2\pi\mu_O}\log|\XI-\E|-h_N^{(j, O, O)}(\XI, \E)\right\}\\
&&+\chi_{\omega^{(j)}}(\XI) \chi_{C\bar{\omega}^{(j)}}(\E) \left\{-\frac{1}{2\pi\mu_O}\log|\XI-\E|-h_N^{(j, I, O)}(\XI, \E)\right\}\\
&&+\chi_{C\bar{\omega}^{(j)}}(\XI) \chi_{\omega^{(j)}}(\E) \left\{-\frac{1}{2\pi\mu_{I_j}}\log|\XI-\E|-h_N^{(j, O, I)}(\XI, \E)\right\}\\
&&+\chi_{\omega^{(j)}}(\XI) \chi_{\omega^{(j)}}(\E) \left\{-\frac{1}{2\pi\mu_{I_j}}\log|\XI-\E|-h_N^{(j, I, I)}(\XI, \E)\right\}\;.
\end{eqnarray*} 
The symmetry of $\cN^{(j)}$, then implies the conditions
\begin{equation}\label{symhN}\left.\begin{array}{c}
\small{\displaystyle{
h_N^{(j, O, O)}(\XI, \E)=h_N^{(j, O, O)}(\E, \XI) \quad \text{ for }\XI, \E \in C\bar{\omega}^{(j)}\;,}}\\ \\
\small{\displaystyle{h_N^{(j, I, O)}(\XI, \E)=h_N^{(j,O, I)}(\E, \XI)
+\frac{1}{2\pi}\left\{\frac{1}{\mu_{I_j}}-\frac{1}{\mu_O}\right\}\log|\XI-\E| \quad \text{ for }\E \in C\bar{\omega}^{(j)}, \XI \in \omega^{(j)}\;, }}\\ \\
\small{\displaystyle{
\text{ and } \quad h_N^{(j, I, I)}(\XI, \E)=h_N^{(j, I, I)}(\E, \XI) \quad \text{ for } \XI, \E \in \omega^{(j)}\;.}}\end{array}\right\}
\end{equation} 
\item We also make use of  model solutions known as  the dipole fields $\cD_k^{(j)}$, $k=1, 2$, $j=1,\dots, M$,  which  play the role of the boundary layers in the asymptotic algorithm. Let $\cD^{(j)}=(\cD^{(j)}_1, \cD^{(j)}_2)^T$, where
\begin{equation*}
\cD^{(j)}(\XI)=\chi_{C\bar{\omega}^{(j)}}(\XI) \cD^{(j,O)}(\XI)+\chi_{\omega^{(j)}}(\XI) \cD^{(j,I)}(\XI)\;,
\end{equation*}
and the  vector functions $\cD^{(j,O)}$, $\cD^{(j,I)}$   solve the  problem
 \begin{equation*}\label{trans8}
\mu_O \Delta_{\XI} \mathcal{D}^{(j,O)}(\XI)=\Oj\;,\quad \XI \in C\bar{\omega}^{(j)}\;,
\end{equation*}
\begin{equation*}\label{trans9}
\mu_{I_j} \Delta_{\XI} \cD^{(j,I)}(\XI)=\Oj\;, \quad \XI \in \omega^{(j)}\;.
\end{equation*}
The  transmission conditions on the boundary of the inclusion $\omega^{(j)}$ are
\begin{equation}\label{trans10}
\mu_{I_j}  \frac{\partial \cD^{(j,I)}}{\partial n_{\XI}}(\XI)-\mu_O  \frac{\partial \cD^{(j,O)}}{\partial n_{\XI}}(\XI)=(\mu_{I_j}-\mu_O) \N^{(j)}\;, \text{ on }\partial \omega^{(j)}\;,
\end{equation}
\begin{equation*}\label{trans11}
 \cD^{(j,O)}(\XI)=\cD^{(j,I)}(\XI)\;, \text{ on } \partial \omega^{(j)}\;,
\end{equation*}
where in (\ref{trans10}), $\N^{(j)}$ is the unit normal to $\omega^{(j)}$. At infinity the vector function $\cD^{(j,O)}$ satisfies
\begin{equation}\label{trans12}
\cD^{(j,O)}(\XI) =O(|\XI|^{-1}), \quad \text{ as }\quad|\XI|\to \infty\;.
\end{equation}

\item We introduce the function $\zeta^{(j)}$ which is a solution of 
\begin{equation*}\label{trans12a1}
\Delta_{\XI} \zeta^{(j)}(\XI)=0\;, \quad \XI \in C\bar{\omega}^{(j)}\;, 
\end{equation*}
\begin{equation*}\label{trans12a2}
\zeta^{(j)}(\XI)=0\;, \quad \XI \in \partial \omega^{(j)}\;,
\end{equation*}
\begin{equation*}\label{trans12a3}
\zeta^{(j)}(\XI) = (2\pi\mu_O)^{-1}\log|\XI|+\zeta^{(j)}_\infty+O(|\XI|^{-1})\;, \quad \text { as }\quad |\XI|\to \infty\;, 
\end{equation*}
where $\zeta^{(j)}_\infty$ is a constant. 

\end{enumerate}

\section{An estimate for solutions to transmission problems for antiplane shear in unbounded domains}\label{transunbest}

The next result  plays an important  role in the asymptotic algorithm. It will allow us to obtain estimates for the boundary layer fields and derive the estimate for a solution of the transmission problem in a domain with multiple inclusions.
 
\begin{lemma}\label{lemunb}
Let $U^{(j)}$ be  a solution of the transmission problem
\begin{equation*}\label{unb1}
\mu_O \Delta U^{(j)}(\XI)=0\;, \quad \XI \in C\bar{\omega}^{(j)}\;,
\end{equation*}
\begin{equation*}\label{unb2}
\mu_{I_j} \Delta U^{(j)}(\XI)=0\;, \quad \XI  \in \omega^{(j)}\;,
\end{equation*}
\begin{equation*}\label{unb3}
U^{(j)}(\XI) \big|_{\XI \in \partial \omega^{(j)+}}=U^{(j)}(\XI)\big|_{\XI\in \partial \omega^{(j)-}}\;,
\end{equation*}
\begin{equation*}\label{unb4}
\mu_O \frac{\partial U^{(j)}}{\partial n}(\XI)\big|_{\XI \in \partial \omega^{(j)+}}-\mu_{I_j} \frac{\partial U^{(j)}}{\partial n}(\XI)\big|_{\XI \in \partial \omega^{(j)-}}=\varphi^{(j)}(\XI)\;,
\end{equation*}
\begin{equation*}\label{unb5}
U^{(j)}(\XI) \to 0 \quad \text{ as }\quad |\XI| \to \infty\;,
\end{equation*}
where $\varphi^{(j)} \in L_\infty(\partial \omega^{(j)})$, $\partial /\partial n$ is the normal derivative on the smooth boundary $\partial \omega^{(j)}$, outward with respect to $\omega^{(j)}$,  and
\begin{equation}\label{un6}
\int_{\partial \omega^{(j)}} \varphi^{(j)}(\XI)\, d S_{\XI}=0\;.
\end{equation}
We also assume that 
\begin{equation*}\label{un7}
\int_{\partial \omega^{(j)}} U^{(j)}(\XI) \frac{\partial \zeta^{(j)}}{\partial n}(\XI) \big|_{\XI \in \partial \omega^{(j)+}} dS_{\XI}=0\;,
\end{equation*}
where $\zeta^{(j)}$ is given as solution of Problem $4$ of Section $\ref{modfield}$.
Then 
\begin{equation}\label{un8}
\sup_{\XI \in C\bar{\omega}^{(j)} \cup \omega^{(j)}} \{ (|\XI|+1) |U^{(j)}(\XI)|\} \le\text{ \rm{const }} \| \varphi^{(j)} \|_{L_\infty(\partial \omega^{(j)})}\;,
\end{equation}
where the constant depends on $\mu_O$, $\mu_{I_j}$ and $\partial \omega^{(j)}$, for  $j=1, \dots, M$.
\end{lemma}

\begin{proof} We note that transmission problems are studied in detail in \cite{CarstSteph} and \cite{CostSteph}, in the context of boundary integral equations and their solvability.

Let us first represent the solution $U^{(j)}$ by two functions $U^{(j, O)}$ and $U^{(j, I)}$, harmonic in the domains $C\bar{\omega}^{(j)}$ and $\omega^{(j)}$, respectively. These functions satisfy
\begin{eqnarray*}\label{bcUO1}
&&U^{(j, O)}(\XI)=U^{(j, I)}(\XI)\;, \quad \XI \in \partial \omega^{(j)}\;,\\
&&\mu_O\frac{\partial U^{(j, O)}}{\partial n}(\XI)-\mu_{I_j}\frac{\partial U^{(j, I)}}{\partial n}(\XI)=\varphi^{(j)}(\XI)\;, \quad \XI \in \partial\omega^{(j)}\;.
\end{eqnarray*}
For the function $U^{(j, O)}$, the condition
\begin{equation}\label{bcUO2}
U^{(j, O)}(\XI) \to 0 \quad \text{ as }\quad |\XI| \to \infty\;,
\end{equation} 
also holds.

Applying Green's formula to the functions $\cN^{(j)}$ (see Problem 2, Section \ref{modfield}), $U^{(j, O)}$ and $U^{(j, I)}$, one obtains 
\begin{equation}\label{intrep1}
U^{(j, O)}(\XI)= -\int_{\partial \omega^{(j)}} \cN^{(j)}(\XI, \E)\varphi^{(j)}(\E)\, dS_{\E}\;,
\end{equation}
for $\XI \in C\bar{\omega}^{(j)}$, and 
\begin{equation}\label{intrep2}
U^{(j, I)}(\XI)= -\int_{\partial \omega^{(j)}} \cN^{(j)}(\XI, \E)\varphi^{(j)}(\E)\, dS_{\E}\;,
\end{equation}
for $\XI \in \omega^{(j)}$.

First, let $|\XI| \ge 2$, then using the  asymptotics for the function $\cN^{(j)}$ at infinity, and  the condition (\ref{un6}), we deduce
\begin{eqnarray}\label{expanatinf1}
(1+|\XI|)|U^{(j, O)}(\XI)| &\le& \text{const }\Big((1+|\XI|)|\log|\XI||\Big|\int_{\partial \omega^{(j)}}\varphi^{(j)}(\E)dS_{\E}\Big|\nonumber\\
&&+\| \varphi^{(j)}\|_{L_1(\partial \omega^{(j)})}\Big)\nonumber\\
&\le&  \text{const}\| \varphi^{(j)}\|_{L_\infty(\partial \omega^{(j)})}\;.
\end{eqnarray}
Also by the Cauchy-Schwarz inequality and (\ref{intrep1})
\begin{equation}\label{expanatinf2}
|U^{(j, O)}(\XI)| \le \text{const }\| \varphi^{(j)}\|_{L_2(\partial \omega^{(j)})}\le \text{const}\| \varphi^{(j)}\|_{L_\infty(\partial \omega^{(j)})}\;, \quad \text{for} \quad \XI \in B_3 \backslash \bar{\omega}^{(j)}\;,
\end{equation}
 where $B_3=\{ \XI: |\XI|<3\}$.

Similarly the integral representation (\ref{intrep2}) gives 
\begin{equation}\label{expanatinf3}
|U^{(j, I)}(\XI)| \le \text{const }\| \varphi\|_{L_2(\partial \omega^{(j)})}\le \text{const}\| \varphi^{(j)}\|_{L_\infty(\partial \omega^{(j)})}\;, \quad \text{for}\quad \XI \in \omega^{(j)}\;.
\end{equation}
The combination of (\ref{expanatinf1}), (\ref{expanatinf2}) and (\ref{expanatinf3}), leads to (\ref{un8}).


\end{proof}

As an immediate corollary of Lemma \ref{lemunb}, we have an estimate for the dipole fields associated with the scaled inclusion $\omega^{(j)}$, $j=1, \dots, M$:

\begin{lemma}\label{transdip}
For the dipole fields $\cD^{(i)}_j$, $j=1,2$, $i=1, \dots, M$,  
\begin{equation*}\label{transdipeq}
\sup_{\XI \in C\bar{\omega}^{(i)} \cup \omega^{(i)}} \{ (|\XI|+1) |\cD^{(i)}_j(\XI)|\} \le \text{\rm{const} }\;, 
\end{equation*}
holds, where the constant in the right-hand side can depend on $\mu_O$, $\mu_{I_i}$ and $\omega^{(i)}$.
\end{lemma}

\section{An estimate for the maximum modulus of solutions to transmission problems for antiplane shear in a domain with several small inclusions}\label{transestmultiplinc}

Here we  obtain an estimate for solutions to transmission problems for antiplane shear in domains with  small inclusions. The next lemma will be used in Section \ref{transmultot} incorporating  the remainder estimates produced by the approximations of Green's function in a perforated domain.

\begin{lemma}\label{multlemtransest}
Let $u$ be a function in $L_\infty(\bigcup_{l=1}^M \omega^{(l)}_\varepsilon\cup\Omega_\varepsilon)$ such that $ \nabla u$ is square integrable in a neighbourhood of $\partial \omega^{(i)}_\varepsilon$, $i=1, \dots, M$. Also, let $u$ be a solution of the transmission problem
\begin{equation}\label{multlemtr1}
\left.
\begin{array}{c}
\displaystyle{\mu_O \Delta u(\X)=0\;, \quad \X \in \Omega_\varepsilon\;,}\\ \\
\displaystyle{\mu_{I_i} \Delta u(\X)=0\;, \quad \X \in \omega^{(i)}_\varepsilon, i=1, \dots, M\;,}\\ \\
\displaystyle{\mu_O \frac{\partial u}{\partial n}(\X)=\psi(\X)\;, \quad \X \in \partial \Omega\;, }\\ \\
\displaystyle{\mu_O \frac{\partial u}{\partial n}(\X)\big|_{\X \in \partial \omega^{(i)+}_\varepsilon}-\mu_{I_i} \frac{\partial u}{\partial n}(\X) \big|_{\X \in \partial \omega^{(i)-}_\varepsilon}=\varphi^{(i)}_\varepsilon(\X)\;,}\\ \\
\displaystyle{u(\X)\big|_{\X \in \partial \omega_\varepsilon^{(i)+}} =u(\X) \big|_{\X \in \partial \omega^{(i)-}_\varepsilon}}\end{array}\right\}
\end{equation}
 where $\psi \in L_\infty(\partial \Omega)$, $\varphi^{(i)}_\varepsilon \in L_\infty (\partial \omega^{(i)}_\varepsilon)$, for $1 \le i \le M$,  
\begin{equation*}\label{multlemtr6}
\int_{\partial \Omega} \psi(\X)\, dS_\X=0\;, \quad \int_{\partial \omega^{(i)}_\varepsilon} \varphi^{(i)}_\varepsilon(\X)\, dS_\X=0\;,
\end{equation*}
and $\varphi_\varepsilon^{(i)}(\X)=\varepsilon^{-1}\varphi^{(i)}(\varepsilon^{-1}(\X-\Oj^{(i)}))$, $i=1, \dots, M$.
To provide uniqueness we also assume 
\begin{equation*}\label{multuniq}
\int_{\partial \Omega} u(\X)dS_\X=0\;.
\end{equation*}

Then there exists a positive constant $A$, independent of $\varepsilon$ and such that 
\begin{equation}\label{multlemtr7}
\| u\|_{L_\infty(\bigcup_{l=1}^M \omega^{(l)}_\varepsilon\cup\Omega_\varepsilon )} \le A \left\{ \| \psi\|_{L_\infty(\partial \Omega)}+\varepsilon\max_{1 \le k \le M}\| \varphi^{(k)}_\varepsilon\|_{L_\infty(\partial \omega_\varepsilon^{(k)})}\right\}\;.
\end{equation}

\end{lemma}

\begin{proof} \emph{ a) The inverse operators to model problems in $\Omega$ and $C\bar{\omega}^{(j)}$, $j=1, \dots, M$. }  Let us introduce the operators 
\begin{equation}\label{lemtr8}
\bcN: \psi \to w\quad \text{ and } \quad \mathfrak{N}^{(j)}: \varphi^{(j)} \to v^{(j)}\;,
\end{equation}
which are the inverse operators of the problems 
\begin{equation}\label{lemtr9}
\left.
\begin{array}{c}
\displaystyle{\mu_O\Delta w(\X)=0\;, \quad \X \in \Omega\;,} \\ \\
\displaystyle{\mu_O \frac{\partial w}{\partial n}(\X)=\psi(\X)\;, \quad \X \in \partial \Omega\;,}\\ \\
\displaystyle{\int_{\partial \Omega} w(\X)\, dS_\X=0\;,}
\end{array}\right\}
\end{equation}
and
\begin{equation}\label{lemtr12}\left.
\begin{array}{c}
\displaystyle{\mu_O \Delta v^{(j)}(\XI)=0\;, \quad \XI \in C\bar{\omega}^{(j)}\;,}\\ \\
\displaystyle{\mu_{I_j} \Delta v^{(j)}(\XI)=0\;, \quad \XI \in \omega^{(j)}\;,}\\ \\
\displaystyle{\mu_O\frac{\partial v^{(j)}}{\partial n}(\XI)\big|_{\XI \in \partial \omega^{(j)+}}-\mu_{I_j} \frac{\partial v^{(j)}}{\partial n}(\XI)\big|_{\XI \in \partial \omega^{(j)-}}=\varphi^{(j)}(\XI)\;,}\\ \\
\displaystyle{v^{(j)}(\XI)\big|_{\XI \in \partial \omega^{(j)+}}=v^{(j)}(\XI)\big|_{\XI \in \partial \omega^{(j)-}}\;,}\\ \\
\displaystyle{v^{(j)}(\XI)\to 0\quad \text{ as }\quad |\XI| \to \infty\;,}
\end{array}\right\}
\end{equation}
where $\psi \in L_\infty(\partial \Omega)$, $\varphi \in L_\infty (\partial \omega^{(j)})$, $j=1, \dots, M$, also
\begin{equation*}\label{lemtr17}
\int_{\partial \omega^{(j)}}\varphi^{(j)}(\XI)\, dS_{\XI}=0\quad \text{ and } \quad \int_{\partial \Omega} \psi(\X)\, dS_\X=0\;.
\end{equation*}

In scaled coordinates $\XI_j=\varepsilon^{-1}(\X-\Oj^{(j)})$, $j=1, \dots, M$, the operator $\mathfrak{N}^{(j)}_\varepsilon$ is defined by
\begin{equation*}\label{lemtr18}
(\mathfrak{N}^{(j)}_\varepsilon \varphi^{(j)}_\varepsilon)(\X)=(\mathfrak{N}^{(j)} \varphi^{(j)})(\XI_j)\;,
\end{equation*}
where $\varphi^{(j)}_\varepsilon(\X)=\varepsilon^{-1} \varphi^{(j)}(\varepsilon^{-1}(\X-\Oj^{(j)}))$.

\vspace{0.1in}\emph{ b)  An estimate for solutions to the model Neumann problem in $\Omega$. } Let $N^{(\Omega)}$ denote the Neumann function $(\ref{defN})$ in $\Omega$.

Then an application of Green's formula to $N^{(\Omega)}(\X, \Y)$ and $w(\Y)$ yields the representation for $w$
\begin{equation}\label{lemtr22}
w(\X)= \int_{\partial \Omega} N^{(\Omega)}(\Y, \X) \psi(\Y)\, dS_\Y +\frac{1}{|\partial\Omega|}\int_{\partial \Omega} w(\Y)\, dS_\Y\;.
\end{equation}
The solution $w$ of the Neumann problem in $\Omega$ is subject to the orthogonality condition in problem  (\ref{lemtr9}), and hence the last term on the right-hand side of (\ref{lemtr22}) is zero, so that 
\begin{equation*}\label{lemtr23}
w(\X)=\int_{\partial \Omega} N^{(\Omega)}(\Y, \X) \psi(\Y)\, dS_\Y\;.
\end{equation*}
From this we obtain the estimate 
\begin{equation}\label{lemtr24}
\sup_{\Omega} |w| \le \text{const } \sup_{\partial \Omega} |\psi|\;.
\end{equation}

\vspace{0.1in}\emph{ c) The case of the homogeneous boundary condition on $\partial \Omega$. }
When the right-hand side of the Neumann condition on $\partial \Omega$ in (\ref{multlemtr1}) is zero we look for a solution of the form
\begin{equation}\label{multlemeq2}
u_1=\sum^M_{j=1} \mathfrak{N}^{(j)}_\varepsilon g^{(j)}_\varepsilon-\bcN\Big(\text{Tr}_{\partial \Omega} \,\mu_O \sum^M_{j=1} \frac{\partial}{\partial n}(\mathfrak{N}^{(j)}_\varepsilon g^{(j)}_\varepsilon)\Big)\;,
\end{equation}
where $g^{(j)}_\varepsilon$ is an unknown function defined on $\partial \omega^{(j)}_\varepsilon$ such that 
\begin{equation*}\label{multlemeq3}
\int_{\partial \omega^{(j)}} g^{(j)}(\XI_j)\, dS_{\XI_j}=0\;,
\end{equation*}
and $g^{(j)}(\XI_j)=\varepsilon g^{(j)}_\varepsilon(\X)$. The function $u_1$ is harmonic inside $\bigcup_{l=1}^M \omega^{(l)}_\varepsilon \cup \Omega_\varepsilon$, 
and is continuous across the boundaries of the small inclusions. On $\partial \Omega$, $u_1$ satisfies
\begin{equation*}\label{multlemeq4}
\mu_O \frac{\partial u_1}{\partial n}(\X)=0\;.
\end{equation*}
Computing the jump in tractions of $u_1$ on $\partial \omega^{(m)}_\varepsilon$, we obtain
\begin{equation*}\label{multlemeq5}
\varphi^{(m)}_\varepsilon(\X)=\mu_O \frac{\partial u_1}{\partial n}(\X) \big|_{\X \in \partial \omega^{(m)+}_\varepsilon}-\mu_{I_m} \frac{\partial u_1}{\partial n}(\X)\big|_{\X \in \partial \omega^{(m)-}_\varepsilon}=g^{(m)}_\varepsilon+(S^{(m)}_\varepsilon g_\varepsilon)(\X)\;.
\end{equation*}
where $g_\varepsilon(\X)=(g^{(1)}(\X), \dots, g^{(M)}(\X))^T$ and
\begin{eqnarray}\label{multlemeq6}
(S^{(m)}_\varepsilon g_\varepsilon)(\X)&=&(\mu_O-\mu_{I_m}) \Big[\frac{\partial}{\partial n}\Big\{\sum_{\substack{j \ne m\\ 1 \le j \le M}} \mathfrak{N}^{(j)}_\varepsilon g^{(j)}_\varepsilon \Big\}\big|_{\X \in \partial \omega^{(m)}_\varepsilon}\nonumber
\\&&- \frac{\partial}{\partial n}\Big\{\bcN \Big( \text{Tr}_{\partial \Omega}\, \mu_O \sum^M_{j=1} \frac{\partial}{\partial n}(\mathfrak{N}^{(j)}_\varepsilon g_\varepsilon^{(j)})\Big) \Big\}\big|_{\X \in \partial \omega^{(m)}_\varepsilon}\Big]\nonumber\\
\end{eqnarray}

Let $B^{(m)}$ denote a disk centered at $\Oj^{(m)}$ and containing $\omega^{(m)}_\varepsilon$. Using a local estimate for solutions of Laplace's equation, along with Lemma \ref{lemunb} and the definition of $g^{(j)}_\varepsilon$ above, we obtain
\begin{eqnarray}\label{multlemeq7}
\Big\|\frac{\partial}{\partial n}\Big\{\sum_{\substack{j \ne m\\ 1 \le j \le M}} \mathfrak{N}^{(j)}_\varepsilon g^{(j)}_\varepsilon \Big\}\Big\|_{L_\infty(\partial \omega^{(m)}_\varepsilon)}&\le & \text{const } \Big\|\sum_{\substack{j \ne m\\ 1 \le j \le M}} \mathfrak{N}^{(j)}_\varepsilon g^{(j)}_\varepsilon\Big\|_{L_\infty(B^{(m)})}\nonumber\\ 
&\le &  \text{const }\sum_{\substack{j \ne m\\ 1 \le j \le M}}\varepsilon^2\| g^{(j)}_\varepsilon\|_{L_\infty(\partial \omega^{(j)}_\varepsilon)}\;.
\end{eqnarray}
Then, from the  local estimate for harmonic functions, 
 we can also assert the inequality
\begin{eqnarray*}\label{multlemeq8}
&&\Big\|\frac{\partial}{\partial n}\Big\{\bcN \Big( \text{Tr}_{\partial \Omega}\, \mu_O \sum^M_{j=1} \frac{\partial}{\partial n}(\mathfrak{N}^{(j)}_\varepsilon g_\varepsilon^{(j)})\Big) \Big\}\Big\|_{L_\infty(\partial \omega^{(m)}_\varepsilon)} \\&\le&
\text{const } \Big\|\bcN \Big( \text{Tr}_{\partial \Omega}\, \mu_O \sum^M_{j=1} \frac{\partial}{\partial n}(\mathfrak{N}^{(j)}_\varepsilon g_\varepsilon^{(j)})\Big) \Big\|_{L_\infty(B^{(m)})} \\ &\le& \text{const }\sum^M_{j=1}\Big\|  \frac{\partial}{\partial n}(\mathfrak{N}^{(j)}_\varepsilon g_\varepsilon^{(j)}) \Big\|_{L_\infty(\partial \Omega)}  
 \le \text{const }\sum_{j =1}^M\varepsilon^2\| g^{(j)}_\varepsilon\|_{L_\infty(\partial \omega^{(j)}_\varepsilon)}\;,
\end{eqnarray*}
where in moving to the second inequality we used estimate (\ref{lemtr24}), and then Lemma \ref{lemunb} brings us to the last inequality. 

Then, the preceding estimate and (\ref{multlemeq6}), (\ref{multlemeq7})  lead to 
\begin{equation*}\label{multlemeq9}
\| S^{(m)}_\varepsilon g_\varepsilon\|_{L_\infty(\partial \omega^{(m)}_\varepsilon)} \le \text{const } \varepsilon^2 \sum^M_{j=1}\| g^{(j)}_\varepsilon\|_{L_\infty(\partial \omega^{(j)}_\varepsilon)}\;.
\end{equation*}
Hence, from the smallness of $S_\varepsilon^{(m)}$, we can write
\begin{equation*}\label{multlemeq10}
g_\varepsilon=(\mathbf{I}+\mathbf{S}_\varepsilon)^{-1}\varphi_\varepsilon\;,
\end{equation*}
where $g_\varepsilon(\X)=(g^{(1)}_\varepsilon(\X), \dots, g^{(M)}_\varepsilon(\X))^T$, $\varphi_\varepsilon(\X)=(\varphi^{(1)}_\varepsilon(\X), \dots, \varphi^{(M)}_\varepsilon(\X))^T$, $\mathbf{S}_\varepsilon$ is a matrix whose rows are $S_\varepsilon^{(1)}, \dots,S_\varepsilon^{(M)}$, and
\begin{equation}\label{multlemeq11}
\| g_\varepsilon^{(j)}\|_{L_\infty(\partial \omega^{(j)}_\varepsilon)} \le \text{const } \max_{1 \le k \le M}\| \varphi^{(k)}_\varepsilon\|_{L_\infty(\partial \omega^{(k)}_\varepsilon)}\;.
\end{equation}
From (\ref{multlemeq2}), together with (\ref{lemtr24}) and Lemma \ref{lemunb}, we obtain
\begin{equation*}\label{multlemeq12}
 \|u_1\|_{L_\infty(\bigcup_{l=1}^M \omega^{(l)}_\varepsilon \cup \Omega_\varepsilon)}\le \text{const } \sum^M_{j=1}\varepsilon \| g^{(j)}_\varepsilon\|_{L_\infty(\partial \omega^{(j)}_\varepsilon)}\;.
\end{equation*}
Now,  this and (\ref{multlemeq11}) give
\begin{equation}\label{multlemeq13}
\|u_1\|_{L_\infty(\bigcup_{l=1}^M \omega^{(l)}_\varepsilon \cup \Omega_\varepsilon)} \le \text{const }\varepsilon \max_{1 \le k \le M}\|\varphi_\varepsilon^{(k)}\|_{L_\infty(\partial \omega^{(k)}_\varepsilon)}\;.
\end{equation} 

\emph{ d) The case of continuous tractions on $\partial \omega^{(j)}_\varepsilon$, $j=1, \dots, M$. }  In this situation we look for the solution $u_2$ in the form
\begin{equation}\label{multlemeq14}
u_2=\bcN\psi+v\;.
\end{equation}
Then, $v$ is a solution of (\ref{multlemtr1}) with the boundary conditions
\begin{equation*}\label{multlemtr15}
\mu_O \frac{\partial v}{\partial n}(\X)=0\;, \text{ on }\partial \Omega\;, 
\end{equation*}
and
\begin{equation*}\label{multlemtr16}
v(\X)\big|_{\X \in \partial \omega^{(j)+}_\varepsilon}=v(\X)\big|_{\X \in \partial \omega^{(j)-}_\varepsilon}
\end{equation*}
\begin{equation*}\label{multlemtr17}
\mu_O\frac{\partial v}{\partial n}(\X)\big|_{\X \in \partial \omega^{(j)+}_\varepsilon}-\mu_{I_j}\frac{\partial v}{\partial n}(\X)\big|_{\X \in \partial \omega^{(j)-}_\varepsilon}=(\mu_{I_j}-\mu_O)\frac{\partial \bcN\psi}{\partial n}(\X)\big|_{\X \in \partial \omega^{(j)}_\varepsilon}
\end{equation*}
for $1 \le j \le M$, where the right-hand sides of the above traction conditions on $\partial \omega_\varepsilon^{(j)}$, $j=1, \dots, M$ are self balanced, so by part c) of the present proof
\begin{eqnarray*}\label{multlemtr18}
 \|v\|_{L_\infty(\bigcup_{l=1}^M \omega^{(l)}_\varepsilon \cup \Omega_\varepsilon)}&\le &\text{const }\varepsilon \max_{1 \le k\le M}\big\|\frac{\partial \bcN\psi}{\partial n}(\X)\big\|_{L_\infty(\partial \omega^{(k)}_\varepsilon)} \nonumber\\ 
&\le& \text{const }\varepsilon \max_{1 \le k\le M}\big\| \bcN\psi \big\|_{L_\infty(B^{(k)})}\nonumber\\ 
&\le &\text{const }\varepsilon \| \psi\|_{L_\infty(\partial \Omega)}\;.
\end{eqnarray*}
This inequality, (\ref{lemtr24}) and  (\ref{multlemeq14}), give 
\begin{equation}\label{multlemtr19}
 \|u_2\|_{L_\infty(\bigcup_{l=1}^M \omega^{(l)}_\varepsilon \cup \Omega_\varepsilon)} \le \text{const }\| \psi \|_{L_\infty(\partial \Omega)}\;.
\end{equation}
Finally, we obtain (\ref{multlemtr7}) through the combination of (\ref{multlemeq13}) and (\ref{multlemtr19}).
\end{proof}


\section{An asymptotic approximation  for the regular part 
 of the Green's function $\cN^{(j)}$ at infinity}\label{transregest}

In this section, we shall prove a result which concerns the asymptotics of $h_N^{(j, O)}$ and $h_N^{(j,I)}$, $j=1, \dots, M$,  (see (\ref{trans7})), at infinity.

\begin{lemma}\label{lemesthn}
For $|\XI|>2$. Then

a) 
\begin{equation*}
h_N^{(j,O)}(\XI, \E)=-\cD^{(j, O)}(\E) \cdot \nabla_{\XI}((2\pi \mu_O)^{-1}\log|\XI|^{-1})+O(|\XI|^{-2}(|\E|+1)^{-1})\;,
\end{equation*}
for  $\E\in C\bar{\omega}^{(j)}$,
 
b) 
\begin{eqnarray*}
h_N^{(j, I)}(\XI, \E)&=&-(\mu_{I_j}^{-1}-\mu_O^{-1})\{(2\pi)^{-1}\log|\XI|-\E \cdot \nabla_{\XI}((2\pi)^{-1}\log|\XI|)\}\nonumber\\
 &&-\cD^{(j, I)}(\E) \cdot \nabla_{\XI}((2\pi\mu_O)^{-1}\log|\XI|^{-1})+O(|\XI|^{-2})\;.
\end{eqnarray*}
for $\E\in \omega^{(j)}$.
\end{lemma}


\begin{proof}  First we study the asymptotic behaviour of the functions $h^{(j, O, O)}_N$, $h_N^{(j, I, O)}$, $h_N^{(j, O, I)}$ introduced in Problem 2 of Section \ref{modfield}, in the neighbourhood of infinity. We show that for $|\XI|>2$,

a) 
\begin{equation}\label{hnest}
h_N^{(j,O, O)}(\XI, \E)=-\cD^{(j, O)}(\E) \cdot \nabla_{\XI}((2\pi \mu_O)^{-1}\log|\XI|^{-1})+O(|\XI|^{-2}(|\E|+1)^{-1})\;,
\end{equation}
 for $\E\in C\bar{\omega}^{(j)}$,

b) 
\begin{eqnarray}\label{hnest2}
h_N^{(j, I, O)}(\E, \XI)=-\cD^{(j, I)}(\E) \cdot \nabla_{\XI}((2\pi \mu_O)^{-1}\log|\XI|^{-1})+O(|\XI|^{-2})\;,
\end{eqnarray}
for $\E\in \omega^{(j)}$,

c) \begin{eqnarray}\label{hncoroest1v1}
 h_N^{(j, O, I)}(\XI, \E)&=&-(\mu_{I_j}^{-1}-\mu_O^{-1})\{(2\pi)^{-1}\log|\XI|-\E \cdot \nabla_{\XI}((2\pi)^{-1}\log|\XI|)\}\nonumber\\
 &&-\cD^{(j, I)}(\E) \cdot \nabla_{\XI}((2\pi\mu_O)^{-1}\log|\XI|^{-1})+O(|\XI|^{-2})\;,
\end{eqnarray}
for $\E \in \omega^{(j)}$.

$i)$ \emph{ Auxiliary functions $h^{(j)}_N$ and $\Upsilon^{(j)}$. } Let us introduce the function $h_N^{(j)}$ as
\begin{equation}\label{hneq}
h_N^{(j)}(\XI, \E)=\chi_{C\bar{\omega}^{(j)}}(\E)h_N^{(j,O)}(\XI, \E)+\chi_{\omega^{(j)}}(\E)h_N^{(j, I)}(\XI, \E)\;.
\end{equation}
From Problem 2 of Section \ref{modfield}, when $\E \in C \bar{\omega}^{(j)}\cup \omega^{(j)}$, $h_N^{(j)}$  is a solution of the transmission
problem
\begin{equation*}\label{regest1}
\mu_O \Delta_{\XI} h_N^{(j)}(\XI, \E)=0\;,\quad \XI \in C\bar{\omega}^{(j)}\;,
\end{equation*}
\begin{equation*}\label{regest2}
\mu_{I_j} \Delta_{\XI} h_N^{(j)}(\XI, \E)=0\;, \quad \XI \in \omega^{(j)}\;,
\end{equation*}
\begin{eqnarray}\label{regest3}
&&\mu_{I_j}  \frac{\partial h_N^{(j)}}{\partial n_{\XI}}(\XI, \E)\big|_{\XI \in \partial \omega^{(j)-}}-\mu_O  \frac{\partial h_N^{(j)}}{\partial n_{\XI}}(\XI, \E)\big|_{\XI \in \partial \omega^{(j)+}}\nonumber\\
&=&-\frac{(\mu_{I_j}-\mu_O)}{2\pi}\left(\frac{\chi_{C\bar{\omega}^{(j)}}(\E)}{\mu_O}+\frac{\chi_{\omega^{(j)}}(\E)}{\mu_{I_j}}\right)\frac{\partial}{\partial n_{\XI}}(\log|\XI-\E|)\;, 
\end{eqnarray}
\begin{equation}\label{regest4}
 h_N^{(j)}(\XI, \E)\big|_{\XI \in \partial \omega^{(j)+}}=h_N^{(j)}(\XI, \E)\big|_{\XI \in \partial \omega^{(j)-}}\;,
\end{equation}
\begin{equation*}\label{regest5}
h_N^{(j)}(\XI, \E) = -(\mu_{I_j}^{-1}-\mu_O^{-1})\chi_{\omega^{(j)}}(\E)(2\pi)^{-1}\log|\XI|+O(|\XI|^{-1}), \quad \text{ as }\quad|\XI|\to \infty\;.
\end{equation*}

 We introduce one more auxiliary vector function $\Upsilon^{(j)}(\XI)=\{ \Upsilon_i^{(j)}(\XI)\}^2_{i=1}$, defined by
\begin{equation}\label{regest6}
\Upsilon^{(j)}(\XI)=\XI -\cD^{(j)}(\XI)\;.
\end{equation}
It solves the transmission problem
\begin{equation}\label{regest7}
\left. \begin{array}{c}
\displaystyle{\mu_O \Delta_{\XI} \Upsilon^{(j)}(\XI)=\Oj\;,\quad \XI \in C\bar{\omega}^{(j)}\;,}\\ \\
\displaystyle{\mu_{I_j} \Delta_{\XI} \Upsilon^{(j)}(\XI)=\Oj\;, \quad \XI \in \omega^{(j)}\;,}\\ \\
\displaystyle{\mu_{I_j}  \frac{\partial \Upsilon^{(j)}}{\partial n_{\XI}}(\XI)\big|_{\XI \in \partial \omega^{(j)-}}=\mu_O  \frac{\partial \Upsilon^{(j)}}{\partial n_{\XI}}(\XI)\big|_{\XI \in \partial \omega^{(j)+}}\;, }\\ \\
\displaystyle{\Upsilon^{(j)}(\XI)\big|_{\XI \in \partial \omega^{(j)+}}=\Upsilon^{(j)}(\XI)\big|_{\XI \in \partial \omega^{(j)-}}\;,}\\ \\
\displaystyle{\Upsilon^{(j)}(\XI)= \XI+O(|\XI|^{-1}), \quad \text{ as }\quad|\XI|\to \infty\;,}\end{array}\right\}
\end{equation}
which is consistent with Problem 3 of Section \ref{modfield}. One can also represent $\Upsilon^{(j)}$, as 
\[ \Upsilon^{(j)}(\XI)=\chi_{C\bar{\omega}^{(j)}}(\XI)\Upsilon^{(j,O)}(\XI)+\chi_{\omega^{(j)}}(\XI)\Upsilon^{(j,I)}(\XI)\;.\]

\vspace{0.1in} $ii)$ \emph{The asymptotics of $h_N^{( j, O, O)}$ at infinity. } 
 For $|\XI|>2$, $\E \in C\bar{\omega}^{( j)}$, we have from Lemma 2 in \cite{9}, that the function $h_N^{( j,O,O)}$ defined in (\ref{hneq}), has the asymptotic representation
\begin{equation}\label{regest18}
h_N^{( j, O, O)}(\XI, \E)=\mathcal{C}^{(j,O)}(\E)\cdot \nabla_{\XI}((2\pi )^{-1}\log|\XI|^{-1})+r^{( j, O, O)}_N(\XI, \E)\;,
\end{equation}
where the remainder $r^{(j, O, O)}_N(\XI, \E)$
satisfies 
\begin{equation*}
|r_N^{(j, O, O)}(\XI, \E)| \le \text{const }(1+|\E|)^{-1} |\XI|^{-2}\quad \text{ for } |\XI|>2, \E \in C\bar{\omega}^{(j)}\;.
\end{equation*}
The vector function $\mathcal{C}^{(j,O)}(\E)=\{\mathcal{C}^{(j,O)}_i(\E)\}_{i=1}^2$ is evaluated below.

\vspace{0.1in}\emph{Evaluation of $\mathcal{C}^{(j,O)}(\E)$. }Let $B_R$ be a disk centered at the origin with a sufficiently large radius $R$. We apply Green's formula to $h_N^{(j)}$ and $\Upsilon_i^{(j)}$ in the domain  $B_R \backslash \bar{\omega}^{(j)} \cup \omega^{(j)}$ 
\begin{eqnarray}
0&=&\mu_O \int_{B_R\backslash \bar{\omega}^{(j)}}\left\{ \Upsilon_i^{(j)}(\XI) \Delta_{\XI} h_N^{(j)}(\XI, \E)-h_N^{(j)}(\XI, \E) \Delta_{\XI} \Upsilon_i^{(j)}(\XI)\right\}d\XI\nonumber\\
&&+\mu_{I_j} \int_{\omega^{(j)}}\left\{ \Upsilon_i^{(j)}(\XI) \Delta_{\XI} h_N^{(j)}(\XI, \E)-h_N^{(j)}(\XI, \E) \Delta_{\XI} \Upsilon_i^{(j)}(\XI)\right\}d\XI\nonumber\\ 
&=&\mu_O \int_{\partial B_R}\left\{ \Upsilon_i^{(j)}(\XI) \frac{\partial h_N^{(j)}}{\partial n_{\XI}}(\XI, \E)-h_N^{(j)}(\XI, \E) \frac{\partial \Upsilon_i^{(j)}}{\partial n_{\XI}}(\XI)\right\}dS_{\XI}\nonumber\\ 
&&+\int_{\partial \omega^{(j)}}\left[ \Upsilon_i^{(j)}(\XI)\left\{\mu_{I_j}\frac{\partial h_N^{(j)}}{\partial n_{\XI}}(\XI, \E)\big|_{\partial \omega^{(j)-}}-\mu_O\frac{\partial h_N^{(j)}}{\partial n_{\XI}}(\XI, \E)\big|_{\partial \omega^{(j)+}}\right\}\right.\nonumber\\ 
&&\left. -h_N^{(j)}(\XI, \E)\left\{\mu_{I_j}\frac{\partial \Upsilon_i^{(j)}}{\partial n_{\XI}}(\XI)\big|_{\partial \omega^{(j)-}}-\mu_O\frac{\partial \Upsilon_i^{(j)}}{\partial n_{\XI}}(\XI)\big|_{\partial \omega^{(j)+}}\right\}\right]dS_{\XI}\;, \label{regest12}
\end{eqnarray}
where while  combining the integrals over $\partial \omega^{(j)+}$ and $\partial \omega^{(j)-}$ we have used the continuity conditions  for the functions $h_N^{(j)}$ and $\Upsilon_i^{(j)}$, respectively (see (\ref{regest4}) and problem (\ref{regest7})). 
With the use of the transmission  conditions for $\Upsilon_i^{(j)}$ and $h_N^{(j)}$, when $\E \in C\bar{\omega}^{(j)}$,  this equality becomes
\begin{eqnarray}
0&=&\mu_O\int_{\partial B_R} \left\{ \Upsilon_i^{(j)}(\XI) \frac{\partial h_N^{(j)}}{\partial n_{\XI}}(\XI, \E)-h_N^{(j)}(\XI, \E) \frac{\partial \Upsilon_i^{(j)}}{\partial n_{\XI}}(\XI)\right\}dS_{\XI}\nonumber\\ \nonumber\\ 
&&-\frac{(\mu_{I_j}-\mu_O)}{\mu_O}\int_{\partial \omega^{(j)}} \Upsilon_i^{(j)}(\XI)\frac{\partial}{\partial n_{\XI}}((2\pi)^{-1}\log|\XI-\E|)\, dS_{\XI}\;,\label{regest13}
\end{eqnarray}
for $\E\in C\bar{\omega}^{(j)}$.
The last integral, owing to (\ref{regest6}), is equal to 
\begin{eqnarray}
&& \int_{\partial \omega^{(j)}} \Upsilon_i^{(j)}(\XI)\frac{\partial}{\partial n_{\XI}}((2\pi)^{-1}\log|\XI-\E|)\, dS_{\XI}\nonumber\\
&=& \int_{\partial \omega^{(j)}}\left\{ \xi_i \frac{\partial}{\partial n_{\XI}}((2\pi)^{-1}\log|\XI-\E|)-\cD^{(j)}_i(\XI)\frac{\partial}{\partial n_{\XI}}((2\pi)^{-1}\log|\XI-\E|)\right\}dS_{\XI}\nonumber\\
&&=  \int_{\partial \omega^{(j)}} \left\{(2\pi)^{-1}\log|\XI-\E| \frac{\partial \xi_i}{\partial n_{\XI}}-\cD_i^{(j)}(\XI) \frac{\partial}{\partial n_{\XI}}((2\pi)^{-1}\log|\XI-\E|)\right\}dS_{\XI}\;,\label{regest14}
\end{eqnarray}
where $\E\in C\bar{\omega}^{(j)}$.
Now, using the jump in tractions for the dipole fields, stated in Problem 3 of Section \ref{modfield}, 
we can write
\begin{eqnarray*}
&&\frac{(\mu_{I_j}-\mu_O)}{\mu_O} \int_{\partial \omega^{(j)}} \Upsilon_i^{(j)}(\XI)\frac{\partial}{\partial n_{\XI}}((2\pi)^{-1}\log|\XI-\E|^{-1})\, dS_{\XI}\nonumber\\
&=&\frac{1}{\mu_O}\int_{\partial \omega^{(j)}} \left\{ (2\pi)^{-1} \log |\XI-\E|^{-1} \left[\mu_{I_j}\frac{\partial \cD^{(j,I)}_i}{\partial n_{\XI}}(\XI) -\mu_O \frac{\partial \cD_i^{(j,O)}}{\partial n_{\XI}}(\XI)\right]\right.\nonumber\\ 
&&\left.- (\mu_{I_j}-\mu_O) \cD_i^{(j)}(\XI)  \frac{\partial}{\partial n_{\XI}}((2\pi)^{-1}\log|\XI-\E|^{-1}) \right\}dS_{\XI}\;.
\end{eqnarray*}
Using the continuity of the displacements for the dipole fields (see Problem 3, Section \ref{modfield}), we write the above right-hand side as 
\[\begin{array}{c}
\small{\displaystyle{\frac{\mu_{I_j}}{\mu_O}\int_{\partial \omega^{(j)}} \left\{ (2\pi)^{-1} \log |\XI-\E|^{-1} \frac{\partial \cD^{(j,I)}_i}{\partial n_{\XI}}(\XI)-  \cD^{(j,I)}_i(\XI)  \frac{\partial}{\partial n_{\XI}}((2\pi)^{-1}\log|\XI-\E|^{-1})\right\} dS_{\XI}}}\nonumber\\ 
\small{\displaystyle{-\int_{\partial \omega^{(j)}} \left\{ (2\pi)^{-1} \log |\XI-\E|^{-1} \frac{\partial \cD^{(j,O)}_i}{\partial n_{\XI}}(\XI)-  \cD^{(j,O)}_i(\XI)  \frac{\partial}{\partial n_{\XI}}((2\pi)^{-1}\log|\XI-\E|^{-1})\right\} dS_{\XI}\;.}}\label{regest15}
\end{array}\]
Then, upon applying Green's formula to $\cD^{(j,O)}_i$, $\cD^{(j,I)}_i$   in $B_R \backslash \bar{\omega}^{(j)}$, $\omega^{(j)}$, respectively, with the fundamental solution of $-\Delta$, and using the definition of $\Upsilon^{(j)}$, we  obtain the above in the form
\[\begin{array}{c}
  \displaystyle{\cD^{(j,O)}_i(\E)-\int_{\partial B_R} \Big\{ (2\pi)^{-1} \log |\XI-\E|^{-1} \frac{\partial \cD^{(j,O)}_i}{\partial n_{\XI}}(\XI)}
\\ \\\displaystyle{-  \cD^{(j,O)}_i(\XI)  \frac{\partial}{\partial n_{\XI}}((2\pi)^{-1}\log|\XI-\E|^{-1})\Big\} dS_{\XI}}
\end{array}\]
 where $\E \in C\bar{\omega}^{(j)}$ and by (\ref{trans12}) the integral over $\partial B_R$ decays as $R \to \infty$ like $O(|\log R|/R)$. Therefore, we have 
\begin{equation*}\label{regest16}
\frac{(\mu_{I_j}-\mu_O)}{\mu_O} \int_{\partial \omega^{(j)}} \Upsilon^{(j)}_i(\XI)\frac{\partial}{\partial n_{\XI}}((2\pi)^{-1}\log|\XI-\E|^{-1})\, dS_{\XI}
= \cD^{(j,O)}_i(\E)
\end{equation*}
for $\E \in C\bar{\omega}^{(j)}$.
Equation (\ref{regest13}), with substitution of the preceding equality, leads to 
\begin{equation}\label{regest17}
-\cD^{(j,O)}_i(\E)=\mu_O\int_{\partial B_R} \left\{ \Upsilon_i^{(j)}(\XI) \frac{\partial h_N^{(j)}}{\partial n_{\XI}}(\XI, \E)-h_N^{(j)}(\XI, \E) \frac{\partial \Upsilon_i^{(j)}}{\partial n_{\XI}}(\XI)\right\}dS_{\XI}\;.
\end{equation}
We now aim to determine the vector function $\mathcal{C}^{(j,O)}(\E)=\{\mathcal{C}^{(j, O)}_i(\E)\}_{i=1}^2$ in (\ref{regest18}). Taking the limit $R \to \infty$ in (\ref{regest17}) and using (\ref{regest18}) (where $h_N^{(j, O, O)}(\XI, \E)=h^{(j)}_N(\XI, \E)$  for $\XI,\E \in  C\bar{\omega}^{(j)}$),   we have
\begin{eqnarray*}
- \cD^{(j,O)}_i(\E)&=& \lim_{R \to \infty} \mu_O\mathcal{C}^{(j,O)}(\E)\cdot\int_{\partial B_R} \Big\{\xi_i \frac{\partial}{\partial n_{\XI}}( \nabla_{\XI}((2\pi )^{-1}\log|\XI|^{-1})) \nonumber\\ \nonumber\\
&&-
 \nabla_{\XI}((2\pi )^{-1}\log|\XI|^{-1}) \frac{\partial \xi_i}{\partial n_{\XI}}\Big\}\, dS_{\XI}\nonumber\\ \nonumber\\
&=&\mu_O\mathcal{C}^{(j, O)}_i(\E)\;.
\end{eqnarray*}
Thus, we have derived that
\[ \mathcal{C}^{(j,O)}(\E)=-\frac{\cD^{(j, O)}(\E)}{\mu_O} \;, \]
which corresponds to the leading order term of the function $h^{(j,O,O)}_N$, stated in the current lemma in (\ref{hnest}).

\vspace{0.1in}$iii)$ \emph{The asymptotics of $h^{(j, I, O)}$ at infinity. } We have already shown that for $|\XI| \ge 2$, $\E \in C\bar{\omega}^{(j)}$
\begin{equation}\label{hnOO1} h_N^{(j, O, O)} (\XI, \E)=-\cD^{(j, O)}(\E) \cdot \nabla_{\XI} ((2\pi\mu_O)^{-1}\log|\XI|^{-1})+r_N^{(j, O, O)}(\XI, \E)\;.
\end{equation}
In order to deduce the leading order term of $h^{(j, I, O)}$ at infinity, we recall the relations (\ref{symhN}). First, since $h^{(j, O, O)}_N$ is symmetric for $\XI, \E \in C\bar{\omega}^{(j)}$ we have the above asymptotic representation also holds for $h_N^{(j, O,O)}(\E, \XI)$. Next we allow $\E$ to approach the boundary of the inclusion $ \omega^{(j)}$. For $\E \in \partial \omega^{(j)}$ we have 
\[\cD^{(j, O)}(\E)=\cD^{(j, I)}(\E)\;,\]
\[ h_N^{(j, O, O)}(\E, \XI)=h_N^{(j, I, O)}(\E, \XI).\]
Therefore, allowing $\E \in  \omega^{(j)}$, for $|\XI| > 2$, we arrive at 
\begin{equation}\label{eqIOr} h_N^{(j, I, O)} ( \E, \XI)=-\cD^{(j, I)}(\E) \cdot \nabla_{\XI} ((2\pi\mu_O)^{-1}\log|\XI|^{-1})+r_N^{(j, O, I)}(\XI, \E)\;,
\end{equation}
where $r^{(j, O, I)}_N$ is the remainder term and subject to its smallness the leading order part of (\ref{hnest2}) has been formally deduced.

\vspace{0.1in}\emph{Remainder estimate. }Consider the function
\begin{equation*}
r_N^{(j)}(\XI, \E)= \chi_{C\bar{\omega}^{(j)}}(\E) r^{(j, O, O)}_N(\XI, \E)+\chi_{\omega^{(j)}}(\E) r^{(j, O, I)}_N(\XI, \E)
\end{equation*}
which by (\ref{hnOO1}) and (\ref{eqIOr}) is 
\begin{eqnarray*}
r_N^{(j)}(\XI, \E)&=&\chi_{C\bar{\omega}^{(j)}}(\E)\{h_N^{(j, O, O)} (\E, \XI)+\cD^{(j, O)}(\E) \cdot \nabla_{\XI} ((2\pi\mu_O)^{-1}\log|\XI|^{-1})\}\\
&&+\chi_{\omega^{(j)}}(\E)\{h_N^{(j, I, O)} (\E, \XI)+\cD^{(j, I)}(\E) \cdot \nabla_{\XI} ((2\pi\mu_O)^{-1}\log|\XI|^{-1})\}\;.
\end{eqnarray*}
Let $|\XI|>2$ and  write the problem for $r_N^{(j)}$ with respect to $\E$ as follows
\begin{equation*}\label{regest20}
\mu_O \Delta_{\E} r_N^{(j)}(\XI, \E)=0\;,\quad \E \in C\bar{\omega}^{(j)}\;,
\end{equation*}
\begin{equation*}\label{regest21}
\mu_{I_j} \Delta_{\E} r_N^{(j)}(\XI, \E)=0\;, \quad \E \in \omega^{(j)}\;,
\end{equation*}
\begin{eqnarray}\label{regest22}
&&\mu_{I_j}  \frac{\partial r_N^{(j)}}{\partial n_{\E}}(\XI, \E)\big|_{\E \in \partial \omega^{(j)-}}-\mu_O  \frac{\partial r_N^{(j)}}{\partial n_{\E}}(\XI, \E)\big|_{\E \in \partial \omega^{(j)+}}\nonumber\\
&&=\frac{(\mu_{I_j}-\mu_O)}{\mu_O}\Big\{\frac{\partial}{\partial n_{\XI}}((2\pi)^{-1}\log|\XI-\E|)
-\frac{\partial}{\partial n_{\XI}}((2\pi)^{-1}\log|\XI|)\Big\}\;, 
\end{eqnarray}
\begin{equation*}\label{regest23}
 r_N^{(j)}(\XI, \E)\big|_{\E \in \partial \omega^{(j)+}}=r_N^{(j)}(\XI, \E)\big|_{\E \in \partial \omega^{(j)-}}\;,
\end{equation*}
\begin{equation*}\label{regest24}
r_N^{(j)}(\XI, \E) \to \Oj, \quad \text{ as }\quad|\E|\to \infty\;.
\end{equation*}
 We have that the right-hand side of (\ref{regest22}) is self balanced and
\begin{eqnarray*}
&&\left|\frac{(\mu_{I_j}-\mu_O)}{\mu_O}\left\{\frac{\partial}{\partial n_{\XI}}((2\pi)^{-1}\log|\XI-\E|)-\frac{\partial}{\partial n_{\XI}}((2\pi)^{-1}\log|\XI|)\right\}\right|\;, \nonumber\\
&\le& \text{const } |\E| |\XI|^{-2} \le \text{const } |\XI|^{-2}\;,\label{hnre}
\end{eqnarray*}
where $\E\in \partial \omega^{(j)}$, $|\E| \le 1$ and $|\XI|>2$. Now an application of Lemma \ref{lemunb}, leads to the estimate for $r_N^{(j)}$
\begin{equation*}\label{estrn}
|r_N^{(j)}(\XI, \E)| \le \text{const } |\XI|^{-2}(|\E|+1)^{-1}\;, 
\end{equation*}
for $|\XI|>2$, $\E \in C \bar{\omega}^{(j)}\cup \omega^{(j)}$. 


\emph{iv) The asymptotics of $h_N^{(j, O, I)}$. }We once again refer to (\ref{symhN}), for the relation
\[h_N^{(j,O, I)}(\XI, \E)=h_N^{(j, I, O)}(\E, \XI)-\frac{1}{2\pi}\left\{\frac{1}{\mu_{I_j}}-\frac{1}{\mu_O}\right\}\log|\XI-\E| \quad \text{ for }\XI \in C\bar{\omega}^{(j)}, \E \in \omega^{(j)}\;.\]
For $|\XI|>2$, $ \E\in \omega^{(j)}$, this can be rewritten as
\[h_N^{(j,O, I)}(\XI, \E)=h_N^{(j, I, O)}(\E, \XI)-\frac{1}{2\pi}\left\{\frac{1}{\mu_{I_j}}-\frac{1}{\mu_O}\right\}\{\log|\XI|-\E \cdot \nabla_{\XI}(\log|\XI|)\} +\left(\frac{1}{|\XI|^{2}}\right)\;.\]
By combining this with (\ref{hnest2}) we obtain (\ref{hncoroest1v1}).
\end{proof}
The proof of Lemma \ref{lemesthn} is then completed by applying (\ref{eqhnIOj}), (\ref{hnest}) and (\ref{hncoroest1v1}).
\section{Uniform approximation of $N_\varepsilon$ in a domain with multiple inclusions}\label{transmultot}


The aim of the current section is to present the uniform asymptotic approximation for $N_\varepsilon$ in a domain with several inclusions.

\begin{theorem}\label{thmtransmult}
The  approximation of Green's function  for the 
transmission problem  of antiplane shear in $\bigcup_l \omega^{(l)}_\varepsilon\cup \Omega_\varepsilon\subset \mathbb{R}^2$, is given by
\boxedeqn{\begin{array}{c}
\small{\displaystyle{N_\varepsilon(\X, \Y)=N^{(\Omega)}(\X, \Y)+\sum_{j=1}^M \cN^{(j)}(\XI_j, \E_j)+M(2\pi\mu_O)^{-1} \log(\varepsilon^{-1}|\X-\Y|)}}\\ \\
\small{\displaystyle{+\varepsilon\sum^M_{j=1}\big\{ \cD^{(j)}(\XI_j) \cdot \nabla_\X R^{(\Omega)}(\Oj^{(j)}, \Y)+ \cD^{(j)}(\E_j) \cdot \nabla_\Y R^{(\Omega)}(\X, \Oj^{(j)})\big\}+O(\varepsilon^2)}}
\end{array}\label{transmultres1}}
uniformly for $\X, \Y \in\bigcup_l \omega^{(l)}_\varepsilon\cup \Omega_\varepsilon$.
\end{theorem}

\begin{proof}
\emph{Formal asymptotic algorithm. }First, we give a plausible argument concerning the representation of $N_\varepsilon$. We propose the function $N_\varepsilon$ to  be given in the form
\begin{equation}\label{tranmultfa11}
N_\varepsilon(\X, \Y)=-\frac{1}{2\pi} \left( \frac{\chi_{\Omega_\varepsilon}(\Y)}{\mu_O}+\sum^M_{l=1} \frac{\chi_{\omega_\varepsilon^{(l)}}(\Y)}{\mu_{I_l}}\right)\log|\X-\Y|-R_\varepsilon(\X, \Y)\;,
\end{equation}
where for $\Y \in \bigcup_{l=1}^M \omega_\varepsilon^{(l)} \cup\Omega_\varepsilon$,  $R_\varepsilon$ is a solution of 
\begin{equation}\label{tranmultfa12}
\mu_O \Delta_\X R_\varepsilon(\X, \Y)=0\;, \quad \X \in \Omega_\varepsilon\;, 
\end{equation} 
\begin{equation}\label{tranmultfa13}
\mu_{I_i} \Delta_\X R_\varepsilon(\X, \Y)=0\;, \quad \X \in \omega_\varepsilon^{(i)}\;, i=1,\dots, M \;, 
\end{equation} 
\begin{equation}\label{tranmultfa14}
 \mu_O \frac{\partial R_\varepsilon}{\partial n_\X}(\X, \Y)=-\frac{\mu_O}{2\pi} \left( \frac{\chi_{\Omega_\varepsilon}(\Y)}{\mu_O}+\sum^M_{l=1} \frac{\chi_{\omega_\varepsilon^{(l)}}(\Y)}{\mu_{I_l}}\right)\frac{\partial(\log|\X-\Y|)}{\partial n_{\X}}+\frac{1}{|\partial \Omega|}\;, \quad  \X \in \partial \Omega\;,
\end{equation}
and for $i=1, \dots, M$, $R_\varepsilon$ satisfies the transmission conditions 
\begin{eqnarray}\label{tranmultfa15}
&& \mu_{I_i} \frac{\partial R_\varepsilon}{\partial n_{\X}}(\X, \Y)\big|_{\X \in \partial \omega^{(i)-}_\varepsilon}-\mu_O  \frac{\partial R_\varepsilon}{\partial n_{\X}}(\X, \Y)\big|_{\X \in \partial \omega^{(i)+}_\varepsilon}\nonumber\\
&=&-\frac{\mu_{I_i}-\mu_O}{2\pi} \left( \frac{\chi_{\Omega_\varepsilon}(\Y)}{\mu_O}+\sum^M_{l=1} \frac{\chi_{\omega_\varepsilon^{(l)}}(\Y)}{\mu_{I_l}}\right)\frac{\partial(\log|\X-\Y|)}{\partial n_{\X}}\;,
\end{eqnarray}
\begin{equation}\label{tranmultfa16}  
R_\varepsilon(\X, \Y)\big|_{\X \in \partial \omega^{(i)+}_\varepsilon}= R_\varepsilon(\X, \Y)\big|_{\X \in \partial \omega^{(i)-}_\varepsilon}\;.
\end{equation}
In the above  problem  we have that $R_\varepsilon$ is subject to the orthogonality condition
\begin{equation}\label{orthReps}
\int_{\partial \Omega}R_\varepsilon(\X, \Y)dS_\X=-\frac{1}{2\pi} \left( \frac{\chi_{\Omega_\varepsilon}(\Y)}{\mu_O}+\sum^M_{l=1} \frac{\chi_{\omega_\varepsilon^{(l)}}(\Y)}{\mu_{I_l}}\right)\int_{\partial \Omega}\log|\X-\Y|\, dS_\X\;.
\end{equation}

\emph{First order approximation for $R_\varepsilon$.}
If we allow $\Y$ to be located inside one of the $M$ inclusions, then we assume $\Y \in \omega^{(m)}_\varepsilon$, where $m$ is fixed, $1 \le m \le M$. We first rewrite the boundary conditions for $R_\varepsilon$.

The boundary condition for $R_\varepsilon$ when $\X \in \partial \Omega$, $\Y \in\Omega_\varepsilon \cup \omega^{(m)}_\varepsilon$, is equivalent to 
\begin{equation*}
\begin{array}{c}
\displaystyle{ \mu_O \frac{\partial R_\varepsilon}{\partial n_\X}(\X, \Y)=\chi_{\Omega_\varepsilon}(\Y)\left\{-\frac{1}{2\pi} \frac{\partial(\log|\X-\Y|)}{\partial n_{\X}}+\frac{1}{|\partial \Omega|}\right\}}\\ \\
 \displaystyle{+\chi_{\omega^{(m)}_\varepsilon}(\Y)\left\{-\frac{\mu_O}{2\pi} \frac{\partial}{\partial n_{\X}}\left[\left\{ \frac{1}{\mu_{I_m}}-\frac{1}{\mu_O}\right\}\log|\X-\Y|+\frac{1}{\mu_O}\log|\X-\Y|\right]+\frac{1}{|\partial \Omega|}\right\}\;.}
 \end{array}
\end{equation*}
Using scaled variables we can also rewrite  the transmission conditions (\ref{tranmultfa15}) on $\partial \omega_\varepsilon^{(q)}$, $q \ne m$, as
\[\begin{array}{c}
\displaystyle{ \mu_{I_q} \frac{\partial R_\varepsilon}{\partial n_{\X}}(\X, \Y)\big|_{\X \in \partial \omega^{(q)-}_\varepsilon}-\mu_O  \frac{\partial R_\varepsilon}{\partial n_{\X}}(\X, \Y)\big|_{\X \in \partial \omega^{(q)+}_\varepsilon}}\\ \\
\displaystyle{ =-\chi_{\Omega_\varepsilon}(\Y)\frac{\mu_{I_q}-\mu_O}{2\pi\mu_O} \frac{\partial(\log|\XI_q-\E_q|)}{\partial n_{\X}}}\\ \\
 \displaystyle{+\chi_{\omega^{(m)}_\varepsilon}(\Y)\Big\{-\frac{\mu_{I_q}-\mu_O}{2\pi} \frac{\partial}{\partial n_{\X}}\left[\left(\frac{1}{\mu_{I_m}}-\frac{1}{\mu_O}\right)\log|\X-\Y|+\frac{1}{\mu_O}\log|\XI_q-\E_q|\right]\Big\}}\;,
\end{array}\]
where $\Y \in \Omega_\varepsilon \cup \omega^{(m)}_\varepsilon$.
Finally,
the transmission condition (\ref{tranmultfa15}) on the $m^{th}$ inclusion  becomes
\begin{eqnarray*}
&&\mu_{I_m} \frac{\partial R_\varepsilon}{\partial n_{\X}}(\X, \Y)\big|_{\X \in \partial \omega^{(m)-}_\varepsilon}-\mu_O  \frac{\partial R_\varepsilon}{\partial n_{\X}}(\X, \Y)\big|_{\X \in \partial \omega^{(m)+}_\varepsilon}\\
&=&-\frac{\mu_{I_m}-\mu_O}{2\pi} \left(\frac{\chi_{\Omega_\varepsilon}(\Y)}{\mu_O}+\frac{\chi_{\omega_\varepsilon^{(m)}}(\Y)}{\mu_{I_m}}\right)\frac{\partial(\log|\XI_m-\E_m|)}{\partial n_{\X}}\;,
\end{eqnarray*}
for $\Y \in \Omega_\varepsilon \cup \omega^{(m)}_\varepsilon$.
We therefore have  the representation
\begin{equation}\begin{array}{c}\label{initReps}
\displaystyle{R_\varepsilon(\X, \Y)=\chi_{\Omega_\varepsilon}(\Y)\Big\{ R^{(\Omega)}(\X, \Y)+\sum_{k=1}^M h^{(k, O)}_N(\XI_k, \E_k)\Big\}}\\ 
\\ \displaystyle{ +\chi_{\omega^{(m)}_\varepsilon}(\Y)\Big\{ R^{(\Omega)}(\X, \Y)+h^{(m, I)}_N(\XI_m, \E_m)}\\\\ \displaystyle{+\sum_{\substack{k \ne m\\ 1 \le k \le M}}h^{(k, O)}_N(\XI_k, \E_k)\Big\}+Z_\varepsilon(\X, \Y)}\;.
\end{array}\end{equation}

\emph{The construction  of boundary layer terms.}
The function $Z_\varepsilon$, given in (\ref{initReps}), is harmonic for $\X \in \bigcup_{l=1}^M \omega_\varepsilon^{(l)} \cup\Omega_\varepsilon$, $\Y \in \Omega_\varepsilon \cup \omega^{(m)}_\varepsilon$ and is continuous across the boundaries of the small inclusions. The normal derivative of $Z_\varepsilon$ on the exterior boundary is
\[\begin{array}{c}
\displaystyle{\mu_O \frac{\partial Z_\varepsilon}{\partial n_\X}(\X, \Y)=-\mu_O\chi_{\Omega_\varepsilon}(\Y)\sum_{k=1}^M \frac{\partial h^{(k, O)}_N}{\partial n_\X}(\XI_k, \E_k)}\\ \\
\displaystyle{-\mu_O\chi_{\omega^{(m)}_\varepsilon}(\Y)\frac{\partial}{\partial n_{\X}}\Big[\frac{1}{2\pi}\left( \frac{1}{\mu_{I_m}}-\frac{1}{\mu_O}\right)\log|\X-\Y|}\\ \\
\displaystyle{+h_N^{(m, I)}(\XI_m, \E_m)+\sum_{\substack{k \ne m\\ 1 \le k \le M}}h_N^{(k, O)}(\XI_k, \E_k)\Big]\;,}
\end{array}\]
 for $\X \in \partial \Omega$,  and the traction transmission conditions for $Z_\varepsilon$ on the small inclusions are 
 \begin{eqnarray*}
&&\mu_{I_m} \frac{\partial Z_\varepsilon}{\partial n_{\X}}(\X, \Y)\big|_{\X \in \partial \omega^{(m)-}_\varepsilon}-\mu_O  \frac{\partial Z_\varepsilon}{\partial n_{\X}}(\X, \Y)\big|_{\X \in \partial \omega^{(m)+}_\varepsilon}
\\ &=&-(\mu_{I_m}-\mu_O)(\chi_{\Omega_\varepsilon}(\Y)+\chi_{\omega_\varepsilon^{(m)}}(\Y)) \frac{\partial}{\partial n_{\X}}\Big\{R^{(\Omega)}(\X, \Y)+ \sum_{\substack{k \ne m\\ 1 \le k \le M}} h_N^{(k, O)}(\XI_k, \E_k)\Big\}\;,
\end{eqnarray*}
 and 
 \[ \begin{array}{c}
 \displaystyle{\mu_{I_q} \frac{\partial Z_\varepsilon}{\partial n_{\X}}(\X, \Y)\big|_{\X \in \partial \omega^{(q)-}_\varepsilon}-\mu_O  \frac{\partial Z_\varepsilon}{\partial n_{\X}}(\X, \Y)\big|_{\X \in \partial \omega^{(q)+}_\varepsilon}}
 \\ \\
 \\\displaystyle{=-(\mu_{I_q}-\mu_O)\Big\{\chi_{\Omega_\varepsilon}(\Y)\frac{\partial}{\partial n_{\X}}\Big[R^{(\Omega)}(\X, \Y)
 + \sum_{\substack{k \ne q\\ 1 \le k \le M}} h_N^{(k, O)}(\XI_k, \E_k)\Big]}\\
\displaystyle{+\chi_{\omega^{(m)}_\varepsilon}(\Y) \frac{\partial}{\partial n_{\X}}\Big[\frac{1}{2\pi}\left(\frac{1}{\mu_{I_m}}-\frac{1}{\mu_O}\right)\log|\X-\Y|+R^{(\Omega)}(\X,\Y)}\\ \\ \displaystyle{+h^{(m, I)}_N(\XI_m, \E_m)+\sum_{\substack{k\ne m\\ k \ne q\\ 1 \le k \le M}}h^{(k, O)}_N(\XI_k, \E_k) \Big]\Big\}\;,}
 \end{array}\]
 where $1 \le q \le M$, $q \ne m$ and $\Y \in \Omega_\varepsilon \cup  \omega_\varepsilon^{(m)}$.
By expanding the first order derivatives of the  logarithmic term near $\chi_{\omega^{(m)}_\varepsilon}(\Y)$, about $\Y=\Oj^{(m)}$ up to $O(\varepsilon^2)$, in the exterior boundary condition and the traction condition on $\partial \omega_\varepsilon^{(q)}$, $q \ne m$, we obtain
\[\begin{array}{c}
\displaystyle{\mu_O \frac{\partial Z_\varepsilon}{\partial n_\X}(\X, \Y)=-\mu_O\chi_{\Omega_\varepsilon}(\Y)\sum_{k=1}^M \frac{\partial h^{(k, O)}_N}{\partial n_\X}(\XI_k, \E_k)}
\\\\ \displaystyle{-\mu_O\chi_{\omega^{(m)}_\varepsilon}(\Y)\Big\{\frac{\partial}{\partial n_{\X}}\Big[\frac{1}{2\pi}\left( \frac{1}{\mu_{I_m}}-\frac{1}{\mu_O}\right)\Big(\log|\X-\Oj^{(m)}|}\\ \\ \displaystyle{-\frac{(\Y-\Oj^{(m)})\cdot (\X-\Oj^{(m)})}{|\X-\Oj^{(m)}|^{2}}
\Big)+h_N^{(m, I)}(\XI_m, \E_m)+\sum_{\substack{k \ne m\\ 1 \le k \le M}}h_N^{(k, O)}(\XI_k, \E_k)\Big]+O(\varepsilon^2)\Big\}\;,} 
\end{array}\]
 for  $\X \in \partial \Omega$, and
\[ \begin{array}{c}
 \displaystyle{\mu_{I_q} \frac{\partial Z_\varepsilon}{\partial n_{\X}}(\X, \Y)\big|_{\X \in \partial \omega^{(q)-}_\varepsilon}-\mu_O  \frac{\partial Z_\varepsilon}{\partial n_{\X}}(\X, \Y)\big|_{\X \in \partial \omega^{(q)+}_\varepsilon}}\\  \\
 \displaystyle{=-(\mu_{I_q}-\mu_O)\Big[\chi_{\Omega_\varepsilon}(\Y)\frac{\partial}{\partial n_{\X}}\Big\{R^{(\Omega)}(\X,\Y)+ \sum_{\substack{k \ne q\\ 1 \le k \le M}} h_N^{(k, O)}(\XI_k, \E_k)\Big\}}\\ \\  \displaystyle{+\chi_{\omega^{(m)}_\varepsilon}(\Y)\Big\{ \frac{\partial}{\partial n_{\X}}\Big[\frac{1}{2\pi}\left(\frac{1}{\mu_{I_m}}-\frac{1}{\mu_O}\right)\Big(\log|\X-\Oj^{(m)}| -\frac{(\Y-\Oj^{(m)})\cdot (\X-\Oj^{(m)})}{|\X-\Oj^{(m)}|^{2}}\Big)}\\ \\
 \displaystyle{+R^{(\Omega)}(\X,\Y)+h^{(m, I)}_N(\XI_m, \E_m)+\sum_{\substack{k\ne m\\ k \ne q\\ 1 \le k \le M}}h^{(k, O)}_N(\XI_k, \E_k) \Big]+O(\varepsilon^2)\Big\}\Big]\;,}
 \end{array}\]
 where $\Y \in \Omega_\varepsilon \cup \omega_\varepsilon^{(m)}$.
 Next we apply Lemma \ref{lemesthn}, in order to rewrite the above boundary condition on $\partial \Omega$ for $Z_\varepsilon$ as
 \[\begin{array}{c}\displaystyle{\mu_O \frac{\partial Z_\varepsilon}{\partial n_\X}(\X, \Y)=-\mu_O\chi_{\Omega_\varepsilon}(\Y)\sum_{k=1}^M \frac{\partial }{\partial n_\X}\Big\{\frac{\varepsilon \cD^{(k, O)}(\E_k)}{2\pi\mu_O} \cdot \frac{\X-\Oj^{(k)}}{|\X-\Oj^{(k)}|^2}
 }\\ 
 \displaystyle{+O\Big( \sum_{k=1}^M \varepsilon^3(|\Y-\Oj^{(k)}|+\varepsilon)^{-1}\Big)\Big\}+\chi_{\omega^{(m)}_\varepsilon}(\Y)\Big\{ -\mu_O\frac{\partial}{\partial n_{\X}}\Big[\frac{\varepsilon\cD^{(m,I)}(\E_m) }{2\pi\mu_O}\cdot \frac{\X-\Oj^{(m)}}{|\X-\Oj^{(m)}|^2}}\\ \\
 \displaystyle{+(\mu_{I_m}^{-1}-\mu_O^{-1})\log\varepsilon+\sum_{\substack{k \ne m\\ 1 \le k \le M}}\frac{\varepsilon \cD^{(k, O)}(\E_k)}{2\pi\mu_O} \cdot \frac{\X-\Oj^{(k)}}{|\X-\Oj^{(k)}|^2}\Big]
 +O(\varepsilon^2)\Big\}\;,}
\end{array}\]
for $\X \in\partial \Omega$, $ \Y \in \Omega_\varepsilon \cup \omega_\varepsilon^{(m)}$. The same lemma  in combination with the Taylor expansion about $\X=\Oj^{(m)}$ of the first order derivatives for the function $R^{(\Omega)}$ leads to 
\[\begin{array}{c}
\displaystyle{\mu_{I_m} \frac{\partial Z_\varepsilon}{\partial n_{\X}}(\X, \Y)\big|_{\X \in \partial \omega^{(m)-}_\varepsilon}-\mu_O  \frac{\partial Z_\varepsilon}{\partial n_{\X}}(\X, \Y)\big|_{\X \in \partial \omega^{(m)+}_\varepsilon}}
\\ \\
\displaystyle{=-(\chi_{\Omega_\varepsilon}(\Y)+\chi_{\omega^{(m)}_\varepsilon}(\Y))(\mu_{I_m}-\mu_O) \N^{(m)} \cdot \nabla_\X R^{(\Omega)}(\Oj^{(m)}, \Y)+O(\varepsilon)\;,\quad \Y \in \Omega_\varepsilon \cup \omega_\varepsilon^{(m)}\;,}
\end{array}\]
where  $\N^{(m)}$ is the unit outward normal to the inclusion $\omega^{(m)}_\varepsilon$.
Similarly, on the $q^{th}$ inclusion, $1 \le q \le M$, $q \ne m$, we derive
\[\begin{array}{c}\displaystyle{\mu_{I_q} \frac{\partial Z_\varepsilon}{\partial n_{\X}}(\X, \Y)\big|_{\X \in \partial \omega^{(q)-}_\varepsilon}-\mu_O  \frac{\partial Z_\varepsilon}{\partial n_{\X}}(\X, \Y)\big|_{\X \in \partial \omega^{(q)+}_\varepsilon}
}\\ \\
\displaystyle{=-(\chi_{\Omega_\varepsilon}(\Y)+\chi_{\omega^{(m)}_\varepsilon}(\Y))(\mu_{I_q}-\mu_O)\N^{(q)} \cdot \nabla_\X R^{(\Omega)}(\Oj^{(q)}, \Y)+O(\varepsilon)}\;, \quad \Y \in \Omega_\varepsilon \cup \omega_\varepsilon^{(m)}\;.\end{array}\]
Then, we approximate $Z_\varepsilon$ by
\begin{equation}\begin{array}{c}
\label{multZeps}
\displaystyle{Z_\varepsilon(\X, \Y)=-(\chi_{\Omega_\varepsilon}(\Y)+\chi_{\omega_\varepsilon^{(m)}}(\Y))  \sum_{\substack{k\ne m \\ 1 \le k \le M}} \varepsilon\{ \cD^{(k, O)}(\E_k) \cdot \nabla_\Y R^{(\Omega)}(\X, \Oj^{(k)})}\\  \\\displaystyle{+\cD^{(k)}(\XI_k)\cdot \nabla_\X R^{(\Omega)}(\Oj^{(k)}, \Y)\}
-\varepsilon\{\chi_{\Omega_\varepsilon}(\Y) \cD^{(m, O)}(\E_m) \cdot \nabla_\Y R^{(\Omega)}(\X, \Oj^{(m)})}\\ \\\displaystyle{+(\chi_{\Omega_\varepsilon}(\Y)+\chi_{\omega^{(m)}_\varepsilon}(\Y))\cD^{(m)}(\XI_m) \cdot \nabla_\X R^{(\Omega)}(\Oj^{(m)}, \Y)\}}
\\ \\\displaystyle{-\chi_{\omega_\varepsilon^{(m)}}(\Y)\{\varepsilon \cD^{(m, I)}(\E_m) \cdot \nabla_\Y R^{(\Omega)}(\X, \Oj^{(m)})+(\mu_{I_m}^{-1}-\mu_O^{-1})\log\varepsilon\}-r^{(m)}_{\varepsilon}(\X, \Y)}\;.\end{array}
\end{equation}

\vspace{0.1in}\emph{Combined formula for $N_\varepsilon$.}
The substitution of (\ref{initReps}), (\ref{multZeps}) into (\ref{tranmultfa11}), for $\Y \in \Omega_\varepsilon \cup \omega^{(m)}_\varepsilon$,  yields
\begin{equation}\begin{array}{c}
\displaystyle{N_\varepsilon(\X, \Y)=\chi_{\Omega_\varepsilon}(\Y)\left\{-(2\pi\mu_O)^{-1} \log|\X-\Y|-R^{(\Omega)}(\X, \Y)-\sum_{k=1}^M h_N^{(k,O)}(\XI_k, \E_k)\right.}\\ \\
\displaystyle{\left.+\varepsilon\sum^M_{k=1}\big\{ \cD^{(k)}(\XI_k) \cdot \nabla_\X R^{(\Omega)}(\Oj^{(k)}, \Y)+ \cD^{(k,O)}(\E_k) \cdot \nabla_\Y R^{(\Omega)}(\X, \Oj^{(k)})\big\}\right\}}
\\ \\
\displaystyle{+\chi_{\omega_\varepsilon^{(m)}}(\Y)\Bigg\{-(2\pi\mu_{I_m})^{-1}\log|\X-\Y|-R^{(\Omega)}(\X,\Y)-h_N^{(m,I)}(\XI_m, \E_m)}\\
\\
\displaystyle{-\sum_{\substack{k \ne m\\ 1 \le k \le M}}h^{(k, O)}_N(\XI_k, \E_k)+(\mu_{I_m}^{-1}-\mu_O^{-1})(2\pi)^{-1}\log \varepsilon
+\varepsilon\, \cD^{(m)}(\XI_m)\cdot \nabla_\X R^{(\Omega)}(\Oj^{(m)}, \Y)}\\
\\
\displaystyle{+\varepsilon \cD^{(m,I)}(\E_m) \cdot \nabla_\Y R^{(\Omega)}(\X, \Oj^{(m)})+\varepsilon \sum_{\substack{k\ne m \\ 1 \le k \le M}}\cD^{(k)}(\XI_k)\cdot \nabla_\X R^{(\Omega)}(\Oj^{(k)}, \Y)}\\ \\
\displaystyle{+\cD^{(k, O)}(\E_k) \cdot \nabla_\Y R^{(\Omega)}(\X, \Oj^{(k)})\Bigg\}+r^{(m)}_\varepsilon(\X, \Y)\;.}\end{array}\label{finNe}
\end{equation}
Then, by the definitions of $N^{(\Omega)}$, $\cN^{(m, O)}$ and $\cN^{(m, I)}$, this is equivalent to
\begin{equation*}\begin{array}{c}
\displaystyle{N_\varepsilon(\X, \Y)=\chi_{\Omega_\varepsilon}(\Y)\left\{N^{(\Omega)}(\X, \Y)+\sum_{k=1}^M \cN^{(k,O)}(\XI_k, \E_k)+M(2\pi\mu_O)^{-1} \log(\varepsilon^{-1}|\X-\Y|)\right.}\\ \\
\displaystyle{\left.+\varepsilon\sum^M_{k=1}\big\{ \cD^{(k)}(\XI_k) \cdot \nabla_\X R^{(\Omega)}(\Oj^{(k)}, \Y)+ \cD^{(k,O)}(\E_k) \cdot \nabla_\Y R^{(\Omega)}(\X, \Oj^{(k)})\big\}\right\}}
\\ \\
\displaystyle{+\sum^M_{j=1}\chi_{\omega_\varepsilon^{(j)}}(\Y)\Bigg\{N^{(\Omega)}(\X,\Y)+\cN^{(j,I)}(\XI_j, \E_j)+\sum_{\substack{k \ne j\\ 1 \le k \le M}}\cN^{(k, O)}(\XI_k, \E_k)}\\ \displaystyle{+M(2\pi\mu_O)^{-1} \log(\varepsilon^{-1}|\X-\Y|)+\varepsilon\, \sum_{k=1}^M\cD^{(k)}(\XI_k)\cdot \nabla_\X R^{(\Omega)}(\Oj^{(k)}, \Y)}\\ \\ \displaystyle{+\varepsilon \cD^{(j,I)}(\E_j) \cdot \nabla_\Y R^{(\Omega)}(\X, \Oj^{(j)})+\varepsilon \sum_{\substack{k\ne j \\ 1 \le k \le M}}\cD^{(k, O)}(\E_k) \cdot \nabla_\Y R^{(\Omega)}(\X, \Oj^{(k)})\Bigg\}+r_\varepsilon(\X, \Y)\;.}\end{array}\label{finNe2}
\end{equation*}

\vspace{0.1in}\emph{Remainder estimates for the approximation of $N_\varepsilon$.}
We represent $r_\varepsilon$ as
\begin{equation}\label{reprepsmult}
r_\varepsilon(\X, \Y)=\chi_{\Omega_\varepsilon}(\Y)\mathfrak{M}_\varepsilon(\X, \Y)+\sum_{j=1}^M \chi_{\omega_\varepsilon^{(p)}}(\Y) \mathfrak{h}^{(j)}_\varepsilon(\X, \Y)\;,
\end{equation}
where $\mathfrak{M}_\varepsilon$ and $\mathfrak{h}^{(j)}_\varepsilon$, $j=1, \dots, M$ are defined below.
In what follows, we estimate $\mathfrak{M}_\varepsilon$, and $\mathfrak{h}^{(j)}_\varepsilon$, $j=1, \dots, M$, in order to estimate $r_\varepsilon$.

\vspace{0.1in}\emph{Remainder estimate for the function $\mathfrak{M}_\varepsilon(\X,\Y)$.}
First, let $\Y \in \Omega_\varepsilon$.
According to (\ref{finNe}), the function $\mathfrak{M}_\varepsilon$ is a solution of 
\begin{equation*}\label{remfmeps1}
\mu_O \Delta_\X \mathfrak{M}_\varepsilon(\X, \Y)=0,\, \quad \X \in \Omega_\varepsilon\;, 
\end{equation*}
\begin{equation*}\label{remfmeps2}
\mu_{I_i} \Delta_\X \mathfrak{M}_\varepsilon(\X, \Y)=0,\, \quad \X \in \omega^{(i)}_\varepsilon, i=1, \dots, M\;, 
\end{equation*}
and on the exterior contour is subject to 
\begin{equation}\label{remfmeps3}
\begin{array}{c}\displaystyle{\mu_O\frac{\partial \mathfrak{M}_\varepsilon}{\partial n_\X}(\X, \Y)=-\varepsilon \mu_O \sum^M_{j=1} \frac{\partial \cD^{(j,O)}}{\partial n_\X}(\XI_j) \cdot \nabla_\X R^{(\Omega)}(\Oj^{(j)}, \Y)}\\ \\ \displaystyle{+\mu_O\sum_{j=1}^M\left\{\frac{\partial h_N^{(j,O)}}{\partial n_\X}(\XI_j, \E_j)-\varepsilon \cD^{(j,O)}(\E_j)\cdot \frac{\partial}{\partial n_\X} \nabla_\Y R^{(\Omega)}(\X, \Oj^{(j)})\right\}, \, \quad \X \in\partial \Omega\;,} 
\end{array}
\end{equation}
with transmission conditions on $\partial \omega_\varepsilon^{(i)}$ of the form
\begin{equation}\begin{array}{c}\label{remfmeps4}
\displaystyle{\mu_{I_i} \frac{\partial \mathfrak{M}_\varepsilon}{\partial n_\X}(\X, \Y) \big|_{\X \in \partial \omega_\varepsilon^{(i)-}}-\mu_O\frac{\partial \mathfrak{M}_\varepsilon}{\partial n_\X}(\X, \Y) \big|_{\X \in \partial \omega_\varepsilon^{(i)+}}}\\ \\
\displaystyle{=(\mu_{I_i}-\mu_O) \frac{\partial R^{(\Omega)}}{\partial n_\X}(\X, \Y)+(\mu_{I_i}-\mu_O) \sum_{\substack{j\ne i\\ 1 \le j \le M}}\frac{\partial h^{(j,O)}_N}{\partial n_\X}(\XI_j, \E_j)}\\ 
\displaystyle{-\varepsilon\sum_{j=1}^M\left\{\mu_{I_i} \frac{\partial \cD^{(j)}}{\partial n_\X}(\XI_j)\big|_{\X\in \partial \omega^{(i)-}_\varepsilon}-\mu_O\frac{\partial \cD^{(j)}}{\partial n_\X}(\XI_j)\big|_{\X\in \partial \omega^{(i)+}_\varepsilon}\right\} \cdot \nabla_\X R^{(\Omega)}(\Oj^{(j)}, \Y)\;,}\\
\displaystyle{-(\mu_{I_i}-\mu_O)\varepsilon \sum^{M}_{j=1} \cD^{(j,O)}(\E_j)\cdot \nabla_\Y\frac{\partial R^{(\Omega)}}{\partial n_\X}(\X, \Oj^{(j)})\;,}
\end{array}
\end{equation}
\begin{equation*}\label{remfmeps5}
 \mathfrak{M}_\varepsilon(\X, \Y)\big|_{\X \in \partial \omega_\varepsilon^{(i)-}}=\mathfrak{M}_\varepsilon(\X, \Y)\big|_{\X \in \partial \omega_\varepsilon^{(i)+}}\;,
\end{equation*}
for $1\le i \le M$, where $\Y \in \Omega_\varepsilon$ and 
\begin{equation}\label{orthM}
 \int_{\partial \Omega}\mathfrak{M}_\varepsilon(\X, \Y)dS_\X=0\;.
\end{equation}
Before estimating the boundary conditions we observe that
\begin{eqnarray}\label{trem61e}
&&\int_{\partial \Omega}  \mu_O \frac{\partial \mathfrak{M}_\varepsilon}{\partial n_\X}(\X, \Y)\, dS_\X=0\;, \\ 
&& \int_{\partial \omega_\varepsilon^{(i)}} \left(\mu_{I_i}  \frac{\partial \mathfrak{M}_\varepsilon}{\partial n_{\X}}(\X, \Y)\big|_{\X \in \partial \omega^{(i)-}_\varepsilon}-\mu_O  \frac{\partial \mathfrak{M}_\varepsilon}{\partial n_{\X}}(\X, \Y)\big|_{\X \in \partial \omega^{(i)+}_\varepsilon}\right)\, dS_\X=0\;,\label{trem61f}
\end{eqnarray}
for $i=1, \dots, M$.

\vspace{0.1in}\emph{Estimation of the right-hand side of $(\ref{remfmeps3})$ on $\partial\Omega$.}
Since $\X \in \partial \Omega$, $|\X-\Oj^{(j)}| \ge 1$ for $j=1, \dots, M$, and by Lemma \ref{transdip}, the estimate 
\begin{equation*}
\varepsilon\mu_O\Big|   \sum^M_{j=1} \frac{\partial \cD^{(j,O)}}{\partial n_\X}(\XI_j) \cdot \nabla_\X R^{(\Omega)}(\Oj^{(j)}, \Y)\Big|\le \text{const }\varepsilon^2\;, 
\end{equation*}
 holds for  $\X \in \partial \Omega, \Y \in \Omega_\varepsilon$. Using Lemma \ref{lemesthn}a), one has for $\X \in \partial \Omega$,
\begin{eqnarray*}\label{remfheps7}
&&\mu_O\Big|\frac{\partial h_N^{(j,O)}}{\partial n_\X}(\XI_j, \E_j)-\varepsilon \cD^{(j,O)}(\E_j)\cdot \frac{\partial}{\partial n_\X} \nabla_\Y R^{(\Omega)}(\X, \Oj^{(j)})\Big| \nonumber\\
&=&\mu_O\Big|\frac{\partial h_N^{(j,O)}}{\partial n_\X}(\XI_j, \E_j)-\frac{\varepsilon \cD^{(j,O)}(\E_j)}{2\pi \mu_O}\cdot \frac{\partial}{\partial n_\X} \Big(\frac{\X-\Oj^{(j)}}{|\X-\Oj^{(j)}|^2}\Big)\Big|\nonumber\\ 
&\le &\text{const }\varepsilon^3|\X-\Oj^{(j)}|^{-2}(|\Y-\Oj^{(j)}|+\varepsilon)^{-1}\le \text{const } \varepsilon^3(|\Y-\Oj^{(j)}|+\varepsilon)^{-1}
\end{eqnarray*}
where we have also made use of the boundary condition $(\ref{trans2})$ for $R^{(\Omega)}$.
The previous two inequalities then give 
\begin{equation}\label{remfmeps6}
\mu_O\Big| \frac{\partial \mathfrak{M}_\varepsilon}{\partial n_\X}(\X, \Y)\Big|\le \text{const }\varepsilon^2\, \quad \X \in \partial \Omega, \Y \in \Omega_\varepsilon\;.
\end{equation}

\vspace{0.1in}\emph{Estimate for the  right-hand side of $(\ref{remfmeps4})$ on $\partial \omega^{(i)}_\varepsilon$, $i=1, \dots, M$.}
The regular part $R^{(\Omega)}$ of $N^{(\Omega)}$ is smooth for $\X , \Y \in \Omega$, and we can expand the first order derivatives of this function about the centre of the small inclusion $\omega^{(i)}_\varepsilon$ to obtain
\begin{eqnarray}\label{remfmeps7}
&&\small{\Big| (\mu_{I_i}-\mu_O) \frac{\partial R^{(\Omega)}}{\partial n_\X} (\X, \Y)
- \varepsilon \Big\{\mu_{I_i} \frac{\partial \cD^{(i,I)}}{\partial n_\X}(\XI_i)\big|_{\X\in \partial \omega^{(i)-}_\varepsilon}}\nonumber\\&&\small{-\mu_O\frac{\partial \cD^{(i,O)}}{\partial n_\X}(\XI_i)\big|_{\X\in \partial \omega^{(i)+}_\varepsilon}\Big\}\cdot \nabla_\X R^{(\Omega)}(\Oj^{(i)}, \Y)\Big| \nonumber}\\
&=&\small{\Big|(\mu_{I_i}-\mu_O) \N^{(i)}\cdot (\nabla_\X R^{(\Omega)}(\X, \Y)-\nabla_\X R^{(\Omega)}(\Oj^{(i)}, \Y))\Big|
\nonumber}\\&\le& \small{\text{const } \varepsilon\, \quad \X \in \partial \omega^{(i)}_\varepsilon, \Y \in \Omega_\varepsilon\;.}
\end{eqnarray}
With the use of Lemma \ref{transdip}, we have
\begin{eqnarray}\label{remfmeps8}
&&\small{\nonumber\varepsilon\Big|\sum_{\substack{j\ne i\\ 1 \le j \le M}}\left\{\mu_{I_i} \frac{\partial \cD^{(j,O)}}{\partial n_\X}(\XI_j)\big|_{\X\in \partial \omega^{(i)-}_\varepsilon}-\mu_O\frac{\partial \cD^{(j,O)}}{\partial n_\X}(\XI_j)\big|_{\X\in \partial \omega^{(i)+}_\varepsilon}\right\} \cdot \nabla_\X R^{(\Omega)}(\Oj^{(j)}, \Y)\Big|}\\\nonumber \\
&&\small{\le  \text{const } \varepsilon^2, \quad \text{ for }\X \in \partial \omega^{(i)}_\varepsilon, \Y \in \Omega_\varepsilon}\;,\nonumber\\ 
\end{eqnarray}
and the same lemma also gives
\begin{eqnarray}\label{essumdE}
&& \Big|\varepsilon \sum^{M}_{j=1} \cD^{(j,O)}(\E_j) \cdot \frac{\partial \nabla_\Y R^{(\Omega)}}{\partial n_\X}(\X, \Oj^{(j)})\Big|\nonumber\\
&\le& \text{const }\sum_{j=1}^M\varepsilon^2 (|\Y-\Oj^{(j)}|+\varepsilon)^{-1}\;,\quad \text{ for } \X \in \partial \omega^{(i)}_\varepsilon, \Y \in \Omega_\varepsilon\;.
\end{eqnarray}
Due to Lemma \ref{lemesthn}a), 
\begin{equation*}\label{estdh}
\Big| \sum_{\substack{j\ne i\\ 1 \le j \le M}}\frac{\partial h^{(j,O)}_N}{\partial n_\X}(\XI, \E)\Big|\le \text{const} \sum_{\substack{j \ne i\\ 1 \le j \le M}}\varepsilon^2(|\Y-\Oj^{(j)}|+\varepsilon)^{-1},\quad \text{ for }\X \in \partial \omega^{(i)}_\varepsilon, \Y \in \Omega_\varepsilon\;.
\end{equation*}
Therefore, this estimate  with (\ref{remfmeps7}), (\ref{remfmeps8}) and (\ref{essumdE}) lead to 
\begin{eqnarray*}\label{remfmeps9}
&&\Big|\mu_{I_i}\frac{\partial \mathfrak{M}_\varepsilon}{\partial n_\X}(\X, \Y) \big|_{\X \in \partial \omega_\varepsilon^{(i)-}}-\mu_O\frac{\partial \mathfrak{M}_\varepsilon}{\partial n_\X}(\X, \Y) \big|_{\X \in \partial \omega_\varepsilon^{(i)+}}\Big| \nonumber\\
&\le&\text{const }\varepsilon, \quad \X \in \partial\omega^{(i)}_\varepsilon, i=1, \dots, M, \Y \in \Omega_\varepsilon\;.
\end{eqnarray*}
Then, by Lemma \ref{multlemtransest} and the preceding estimate with (\ref{orthM})--(\ref{remfmeps6}), we obtain
\begin{equation}\label{remfmeps10}
|\mathfrak{M}_\varepsilon(\X, \Y)| \le \text{const } \varepsilon^2,\quad  \X \in \bigcup^M_{l=1} \omega_\varepsilon^{(l)} \cup \Omega_\varepsilon, \Y \in \Omega_\varepsilon\;.
\end{equation}

\vspace{0.1in}\emph{Remainder estimate for $\mathfrak{h}^{(j)}_\varepsilon$, $j=1, \dots, M$.}
Let $\Y \in \omega_\varepsilon^{(j)}$, where $j$ is fixed, $1 \le j \le M$. Then
by (\ref{transmultres1}), the remainder term $\mathfrak{h}^{(j)}_\varepsilon$  solves
\begin{equation*}\label{piresp1}
\mu_O \Delta_\X \mathfrak{h}^{(j)}_\varepsilon(\X, \Y)=0,\, \quad \X \in \Omega_\varepsilon\;, 
\end{equation*}
\begin{equation*}\label{piresp2}
\mu_{I_i} \Delta_\X \mathfrak{h}^{(j)}_\varepsilon(\X, \Y)=0,\, \quad \X \in \omega^{(i)}_\varepsilon, i=1, \dots, M\;, 
\end{equation*}
and satisfies on the exterior boundary
\begin{equation}
\small{\begin{array}{c}\label{piresp3}
\displaystyle{\mu_O\frac{\partial \mathfrak{h}^{(j)}_\varepsilon}{\partial n_\X}(\X, \Y)=-\mu_O\Bigg\{\frac{\partial \cN^{(j,I)}}{\partial n_\X}(\XI_j, \E_j)+\frac{\partial}{\partial n_\X}((2\pi\mu_O)^{-1}\log(\varepsilon^{-1}|\X-\Y|))}\\ \\ \displaystyle{-\sum_{\substack{k \ne m \\ 1 \le k \le M}}\frac{\partial  h_N^{(k, O)}}{\partial n_\X}(\XI_k, \E_k)+\varepsilon \frac{\partial\cD^{(j,O)}}{\partial n_\X}(\XI_j)\cdot\nabla_\X R^{(\Omega)}(\Oj^{(j)}, \Y)}\\ \\
\displaystyle{+\varepsilon \cD^{(j,I)}(\E_j) \cdot \frac{\partial}{\partial n_\X} \nabla_\Y R^{(\Omega)}(\X, \Oj^{(j)})+\varepsilon \sum_{\substack{k \ne j\\ 1 \le k \le M}}\Big[ \frac{\partial \cD^{(k,O)}}{\partial n_\X}(\XI_k) \cdot \nabla_\X R^{(\Omega)}(\Oj^{(k)}, \Y)}\\ \\
\displaystyle{+\cD^{(k, O)}(\E_k) \cdot \frac{\partial}{\partial n_\X}\nabla_\Y R^{(\Omega)}(\X, \Oj^{(j)})\Big]\Bigg\}, \, \quad \X \in\partial \Omega\;.} 
\end{array}}
\end{equation}
Also, on the boundary of the $i^{th}$ inclusion, $\mathfrak{h}^{(j)}_\varepsilon$ is subject to  
\begin{equation}\small{\begin{array}{c}\label{piresp4}
\displaystyle{\mu_{I_i} \frac{\partial \mathfrak{h}^{(j)}_\varepsilon}{\partial n_\X}(\X, \Y) \big|_{\X \in \partial \omega_\varepsilon^{(i)-}}-\mu_O\frac{\partial \mathfrak{h}^{(j)}_\varepsilon}{\partial n_\X}(\X, \Y) \big|_{\X \in \partial \omega_\varepsilon^{(i)+}}}\\ \\ \displaystyle{
=(\mu_{I_i}-\mu_O) \Bigg[\frac{\partial R^{(\Omega)}}{\partial n_\X} (\X, \Y)-\varepsilon  \cD^{(j,I)}(\E_j) \cdot \frac{\partial }{\partial n_\X}\nabla_\Y R^{(\Omega)}(\X, \Oj^{(j)})}\\ \\ \displaystyle{-\varepsilon \sum_{\substack{k \ne j \\ 1 \le k \le M}}\cD^{(k, O)}(\E_k)\cdot \frac{\partial}{\partial n_\X}\nabla_\Y R^{(\Omega)}(\X, \Oj^{(k)})
+\sum_{\substack{k \ne j\\ 1 \le k \le M}}\Big((1-\delta_{ik}) \frac{\partial h_N^{(k,O)}}{\partial n_\X}(\XI_k, \E_k)} \\ \\\displaystyle{-\delta_{ik} (2\pi\mu_O)^{-1} \frac{\partial}{\partial n_\X}(\log(\varepsilon^{-1}|\X-\Y|))\Big)
-(1-\delta_{ij}) \frac{\partial \cN^{(j,I)}}{\partial n_\X}(\XI_j, \E_j)\Bigg]}\\ \\ 
\displaystyle{-\varepsilon \left\{\mu_{I_i} \frac{\partial \cD^{(j)}}{\partial n_\X}(\XI_j)\big|_{\X \in \partial \omega^{(i)-}_\varepsilon}-\mu_O \frac{\partial \cD^{(j)}}{\partial n_\X}(\XI_j) \big|_{\X \in \partial \omega^{(i)+}_\varepsilon}\right\}\cdot \nabla R^{(\Omega)}(\Oj^{(j)}, \Y)}\\
\\ 
\displaystyle{-\varepsilon \sum_{\substack{k \ne j\\ 1 \le k \le M}}\left\{\mu_{I_i} \frac{\partial \cD^{(k)}}{\partial n_\X}(\XI_k)\big|_{\X \in \partial \omega^{(i)-}_\varepsilon}-\mu_O \frac{\partial \cD^{(k)}}{\partial n_\X}(\XI_k) \big|_{\X \in \partial \omega^{(i)+}_\varepsilon}\right\}\cdot \nabla_\X R^{(\Omega)}(\Oj^{(k)}, \Y)\;,} 
\end{array}}
\end{equation}
\begin{equation*}\label{piresp5}
 \mathfrak{h}^{(j)}_\varepsilon(\X, \Y)\big|_{\X \in \partial \omega_\varepsilon^{(i)-}}=\mathfrak{h}^{(j)}_\varepsilon(\X, \Y)\big|_{\X \in \partial \omega_\varepsilon^{(i)+}}\;,
\end{equation*}
for $i=1, \dots, M$, where $\Y \in \omega^{(j)}_\varepsilon$
and
\begin{equation}\label{orthhj}
 \int_{\partial \Omega}\mathfrak{h}^{(j)}_\varepsilon(\X, \Y)dS_\X=0\;.
 \end{equation}
From this problem we can see
\begin{eqnarray}\label{trem61i}
&&\int_{\partial \Omega}  \mu_O \frac{\partial \mathfrak{h}^{(j)}_\varepsilon}{\partial n_\X}(\X, \Y)\, dS_\X=0\;, \\ 
&& \int_{\partial \omega_\varepsilon^{(j)}} \left(\mu_{I_j}  \frac{\partial \mathfrak{h}^{(j)}_\varepsilon}{\partial n_{\X}}(\X, \Y)\big|_{\X \in \partial \omega^{(j)-}_\varepsilon}-\mu_O  \frac{\partial \mathfrak{h}^{(j)}_\varepsilon}{\partial n_{\X}}(\X, \Y)\big|_{\X \in \partial \omega^{(j)+}_\varepsilon}\right)\, dS_\X=0\;,\label{trem61j}
\end{eqnarray}
for $i=1, \dots, M$. 

Before  estimating the discrepancies in the boundary conditions, we note for $\Y \in \omega^{(j)}_\varepsilon$,  by Lemma \ref{transdip} 
\begin{equation}\label{innerest1}
\Big|\sum_{\substack{k \ne j \\ 1 \le k \le M}}\cD^{(k, O)}(\E_k)\cdot \frac{\partial}{\partial n_\X}\nabla_\Y R^{(\Omega)}(\X, \Oj^{(k)})\Big| \le \text{const }\varepsilon,
\quad\text{ holds for }\X \in \partial \Omega_\varepsilon\;,
\end{equation}
whereas the same lemma in conjunction with Lemma \ref{lemesthn}a), leads to 
\begin{equation}\label{innerest2}
\Big|\frac{\partial h_N^{(k,O)}}{\partial n_\X}(\XI_k, \E_k)\Big| \le \text{const }\varepsilon^2\;,\quad\text{ for } \X \in 
\bigcup_{\substack{i \ne k\\ 1 \le i \le M}} \partial \omega^{(i)}_\varepsilon\cup \partial \Omega\;.
\end{equation}

\vspace{0.1in}\emph{Estimate of the right-hand side of $(\ref{piresp3})$ on $\partial \Omega$.}
The definition of $\cN^{(j, I)}$ in Section \ref{modfield} and  the boundary condition (\ref{trans2}) for $R^{(\Omega)}$, give us
\begin{eqnarray*}
&&\small{\mu_O\Big| \frac{\partial}{\partial n_\X}\Big\{\cN^{(j,I)}(\XI_j, \E_j)+(2\pi\mu_O)^{-1}\log(\varepsilon^{-1}|\X-\Y|)+\varepsilon \cD^{(j,I)}(\E_j) \cdot  \nabla_\Y R^{(\Omega)}(\X, \Oj^{(j)})\Big\}\Big|\nonumber}\\ 
&=&\small{\mu_O\Big|\frac{\partial}{\partial n_\X}\{-(2\pi)^{-1}(\mu_{I_j}^{-1}-\mu_O^{-1})\log|\X-\Y|-h_N^{(j, I)}(\XI_j, \E_j)\nonumber}\\&&\small{+\varepsilon(2\pi \mu_O)^{-1} \cD^{(j,I)}(\E_j) \cdot (\X-\Oj^{(j)})|\X-\Oj^{(j)}|^{-2}\Big|\quad \text{ for } \X \in \partial \Omega, \Y \in \omega^{(j)}_\varepsilon\;.}
\end{eqnarray*}
Using the asymptotics of $h_N^{(j,I)}$ at infinity contained in Lemma \ref{lemesthn}b), we obtain
\begin{eqnarray}\label{piresp6}
&&\small{\mu_O\Big| \frac{\partial}{\partial n_\X}\left\{\cN^{(j,I)}(\XI_j, \E_j)+(2\pi\mu_O)^{-1}\log(\varepsilon^{-1}|\X-\Y|)+\varepsilon \cD^{(j,I)}(\E_j) \cdot  \nabla_\Y R^{(\Omega)}(\X, \Oj^{(j)})\right\}\Big|}\nonumber\\
&\le& \small{\text{const }\varepsilon^2\;, \quad \text{ for } \X \in \partial \Omega, \Y \in \omega^{(j)}_\varepsilon\;.}\nonumber\\
\end{eqnarray}
Then from Lemma \ref{transdip} we have
\begin{eqnarray*}\label{piresp7}
&&\varepsilon \mu_O \Big| \frac{\partial}{\partial n_\X} \Big\{\cD^{(j,O)}(\XI_j)\cdot\nabla_\X R^{(\Omega)}(\Oj^{(j)}, \Y)+ \sum_{\substack{k \ne j\\ 1 \le k \le M}}  \cD^{(k,O)}(\XI_k) \cdot \nabla_\X N^{(\Omega)}(\Oj^{(k)}, \Y) \Big\}\Big|\nonumber \\ 
& \le &\text{const } \varepsilon^2, \quad \X \in \partial \Omega, \Y \in \omega^{(j)}_\varepsilon\;.
\end{eqnarray*}
Therefore, this  along with (\ref{innerest1}), (\ref{innerest2}) and  (\ref{piresp6})  yield
\begin{equation}\label{piresp8}
\mu_O\Big|\frac{\partial \mathfrak{h}^{(j)}_\varepsilon}{\partial n_\X}(\X, \Y)\Big| \le \text{const } \varepsilon^2\;,\quad \X \in \partial \Omega, \Y \in \omega^{(j)}_\varepsilon\;.
\end{equation}

\vspace{0.1in}\emph{Estimate for the right-hand side of $(\ref{piresp4})$ on $\partial \omega^{(i)}_\varepsilon$, $i=1, \dots, M$.} 
Consider first the situation when $i=j$. Since $R^{(\Omega)}$ is smooth for $\X, \Y \in \Omega$, we can take the  Taylor expansion of its first order derivatives about $\X=\Oj^{(j)}$ to obtain
\begin{eqnarray}\label{piresp9}
&&\Big| (\mu_{I_j}-\mu_O) \frac{\partial R^{(\Omega)}}{\partial n_\X} (\X, \Y) 
\nonumber-\varepsilon \Big\{\mu_{I_j} \frac{\partial \cD^{(j)}}{\partial n_\X}(\XI_j)\big|_{\X \in \partial \omega^{(j)-}_\varepsilon}\\ &&-\mu_O \frac{\partial \cD^{(j)}}{\partial n_\X}(\XI_j) \big|_{\X \in \partial \omega^{(j)+}}\Big\}\cdot \nabla R^{(\Omega)}(\Oj^{(j)}, \Y)\Big|\nonumber\\
&=&|\mu_{I_j}-\mu_O| | \N^{(j)} \cdot \nabla_\X R^{(\Omega)}(\X, \Y)-\N^{(j)}\cdot \nabla_\X R^{(\Omega)}(\Oj^{(j)}, \Y)\Big|\nonumber\\ 
& \le & \text{const } \varepsilon, \quad \X \in \partial \omega^{(j)}_\varepsilon, \Y \in \omega^{(j)}_\varepsilon\;,
\end{eqnarray}
where the boundary condition (\ref{trans10}) for the dipole fields of the inclusion $\partial \omega^{(i)}_\varepsilon$ was also used. Then, owing to Lemma \ref{transdip} 
for $i=j$ one has
\begin{eqnarray*}\label{piresp10}
&&\small{\Big| \varepsilon (\mu_{I_j}-\mu_O) \cD^{(j,I)}(\E_j) \cdot \frac{\partial }{\partial n_\X}\nabla_\Y R^{(\Omega)}(\X, \Oj^{(j)})\nonumber} \\ 
&&\small{-\varepsilon \sum_{\substack{k \ne j\\ 1 \le k \le M}}\left\{\mu_{I_j} \frac{\partial \cD^{(k,O)}}{\partial n_\X}(\XI_k)\big|_{\X \in \partial \omega^{(j)-}_\varepsilon}-\mu_O \frac{\partial \cD^{(k,O)}}{\partial n_\X}(\XI_k) \big|_{\X \in \partial \omega^{(j)+}}\right\}\cdot \nabla_\X R^{(\Omega)}(\Oj^{(k)}, \Y)\Big|}
\nonumber \\ 
&\le& \small{\text{const } \varepsilon, \quad \X \in \partial \omega^{(j)}_\varepsilon, \Y \in \omega^{(j)}_\varepsilon\;.}
\end{eqnarray*}
This and (\ref{innerest1}), (\ref{innerest2}), (\ref{piresp9}) lead to the estimate
\begin{equation}\label{piresp10a}
\Big| \mu_{I_j} \frac{\partial \mathfrak{h}^{(j)}_\varepsilon}{\partial n_\X}(\X, \Y) \big|_{\X \in \partial \omega_\varepsilon^{(j)-}}-\mu_O\frac{\partial \mathfrak{h}^{(j)}_\varepsilon}{\partial n_\X}(\X, \Y) \big|_{\X \in \partial \omega_\varepsilon^{(j)+}}\Big| 
\le \text{const }\varepsilon, \quad\Y \in \omega^{(j)}_\varepsilon\;.
\end{equation}
It remains to consider the case $i \ne j$. In this situation, we require Lemma \ref{lemesthn}b) to obtain
\begin{eqnarray}\label{piresp11}
&&
\small{\Big|  \frac{\partial \cN^{(j,I)}}{\partial n_\X}(\XI_j, \E_j) 
+(2\pi\mu_O)^{-1}\frac{\partial}{\partial n_\X}(\log(\varepsilon^{-1}|\X-\Y|))}\nonumber\\
&& \small{+\varepsilon \cD^{(j,I)}(\E_j) \cdot \frac{\partial }{\partial n_\X}\nabla_\Y R^{(\Omega)}(\X, \Oj^{(j)})\Big|}\nonumber\\ 
&=&\small{\varepsilon 
\Big|\cD^{(j,I)}(\E_j) \cdot \Big\{\frac{\partial }{\partial n_\X}\nabla_\X ((2\pi\mu_O)^{-1}\log|\X-\Oj^{(j)}|^{-1})- \frac{\partial }{\partial n_\X}\nabla_\Y R^{(\Omega)}(\X, \Oj^{(j)})\Big\}\Big|}\nonumber\\ 
&\le&\small{\text{const } \varepsilon\;, \quad \X \in \partial \omega^{(i)}_\varepsilon, \Y \in \omega^{(j)}_\varepsilon, i \ne j\;.}\nonumber\\
\end{eqnarray}
The boundary conditions for the dipole fields (\ref{trans10}) give
\begin{eqnarray}\label{piresp12}
&&\Big|(\mu_{I_i}-\mu_O) \frac{\partial R^{(\Omega)}}{\partial n_\X} (\X, \Y)\nonumber\\
&&-\varepsilon\left\{\mu_{I_i} \frac{\partial \cD^{(i)}}{\partial n_\X}(\XI_i)\big|_{\X \in \partial \omega^{(i)-}_\varepsilon}-\mu_O \frac{\partial \cD^{(i)}}{\partial n_\X}(\XI_i) \big|_{\X \in \partial \omega^{(i)+}_\varepsilon}\right\}\cdot \nabla_\X R^{(\Omega)}(\Oj^{(i)}, \Y)\Big|\nonumber\\ 
&\le&|\mu_{I_i}-\mu_O| | \N^{(i)}\cdot  \nabla_\X R^{(\Omega)}(\X, \Y)-\nabla_\X R^{(\Omega)}(\Oj^{(i)}, \Y)|\nonumber\\
&\le&  \text{const } \varepsilon, \quad \X \in \partial \omega^{(i)}_\varepsilon,\Y \in \omega^{(j)}_\varepsilon, i \ne j\;. \nonumber \\
\end{eqnarray}
Then,
Lemma \ref{transdip}, allows one to deduce
\begin{equation*}\label{piresp13}\begin{array}{c}
\displaystyle{\varepsilon\Big|\left\{\mu_{I_i} \frac{\partial \cD^{(j,O)}}{\partial n_\X}(\XI_j)\big|_{\X \in \partial \omega^{(i)-}_\varepsilon}-\mu_O \frac{\partial \cD^{(j,O)}}{\partial n_\X}(\XI_j) \big|_{\X \in \partial \omega^{(i)+}_\varepsilon}\right\}\cdot \nabla_\Y R^{(\Omega)}(\Oj^{(j)}, \Y)}\\  \\
\displaystyle{+\sum_{\substack{k \ne j\\ k \ne i\\ 1 \le k \le M}}\left\{\mu_{I_i} \frac{\partial \cD^{(k,O)}}{\partial n_\X}(\XI_k)\big|_{\X \in \partial \omega^{(i)-}_\varepsilon}-\mu_O \frac{\partial \cD^{(k,O)}}{\partial n_\X}(\XI_k) \big|_{\X \in \partial \omega^{(i)+}_\varepsilon}\right\}\cdot \nabla_\X R^{(\Omega)}(\Oj^{(k)}, \Y) \Big|} \\ \\
\displaystyle{\le \text{const } \varepsilon^2, \quad \X \in \partial \omega^{(i)}_\varepsilon, \Y \in \omega^{(j)}_\varepsilon, i \ne j\;.}
\end{array}
\end{equation*}
Thus, from the preceding inequality and (\ref{innerest1}), (\ref{innerest2}), (\ref{piresp11}), (\ref{piresp12}),  we have 
\begin{eqnarray*}\label{piresp14}
\Big| \mu_{I_i} \frac{\partial \mathfrak{h}^{(j)}_\varepsilon}{\partial n_\X}(\X, \Y) \big|_{\X \in \partial \omega_\varepsilon^{(i)-}}-\mu_O\frac{\partial \mathfrak{h}^{(j)}_\varepsilon}{\partial n_\X}(\X, \Y) \big|_{\X \in \partial \omega_\varepsilon^{(i)+}}\Big|\le \text{const }\varepsilon, \quad \Y \in \omega^{(j)}_\varepsilon, i \ne j\;.
\end{eqnarray*}
We can conclude from this estimate, the conditions (\ref{orthhj})--(\ref{trem61j}), inequalities (\ref{piresp8}),   (\ref{piresp10a}),  and Lemma \ref{multlemtransest} that
\begin{equation*}\label{piresp15}
\Big| \mathfrak{h}^{(j)}_\varepsilon(\X, \Y) \Big| \le \text{const } \varepsilon, \quad \X \in \bigcup^M_{l=1}\omega^{(l)}_\varepsilon \cup \Omega_\varepsilon, \Y \in \omega^{(j)}_\varepsilon, j=1, \dots, M\;.
\end{equation*}
Finally, the above estimate for $\mathfrak{h}^{(j)}_\varepsilon$, $j=1, \dots, M$ and (\ref{remfmeps10}), combined with  (\ref{reprepsmult})  complete the proof.
\end{proof}

\section{Numerical simulations}\label{antiplnumsim}

In the current section, we implement the asymptotic formulae derived in Section \ref{transmultot} for Green's function $N_\varepsilon$ in numerical simulations. 
The numerical computations are carried out for the regular part $R_\varepsilon$ of Green's function in $\bigcup_l \omega^{(l)}_\varepsilon \cup\Omega_\varepsilon$ for the transmission problem. In other words, let this regular part be given by the formula
\begin{equation}\label{nsant1}
R_\varepsilon(\X, \Y)=-\frac{1}{2\pi} \left( \frac{\chi_{\Omega_\varepsilon}(\Y)}{\mu_O}+\sum^M_{l=1} \frac{\chi_{\omega_\varepsilon^{(l)}}(\Y)}{\mu_{I_l}}\right)\log|\X-\Y|-N_\varepsilon(\X, \Y)\;.
\end{equation}

Then in accordance with the boundary value problem for $N_\varepsilon$ given in Section \ref{transnotmult}, the function $R_\varepsilon$, which we choose to consider for our numerical schemes, when $\Y \in \Omega_\varepsilon$  is a solution of the problem (\ref{tranmultfa12})--(\ref{orthReps}).

\subsection{Asymptotic formulae for $R_\varepsilon$}

From formula (\ref{nsant1}) and Theorem \ref{thmtransmult}, we can immediately state the asymptotic formulae for the regular part $R_\varepsilon$ that will be used in the examples below. When $\Y \in \Omega_\varepsilon$ we have  from (\ref{transmultres1}), that $R_\varepsilon$ admits the asymptotic representation
\begin{equation}\begin{array}{c}
\displaystyle{R_\varepsilon(\X, \Y)=R^{(\Omega)}(\X, \Y)+\sum_{j=1}^M h_N^{(j,O)}(\XI_j, \E_j)}\\ \\
\displaystyle{-\varepsilon\sum^M_{j=1}\big\{ \cD^{(j)}(\XI_j) \cdot \nabla_\X R^{(\Omega)}(\Oj^{(j)}, \Y)+ \cD^{(j,O)}(\E_j) \cdot \nabla_\Y R^{(\Omega)}(\X, \Oj^{(j)})\big\}+O(\varepsilon^2)}\;.
\end{array}\label{regptransmultpo}
\end{equation}
When $\Y \in \omega_\varepsilon^{(j)}$ where $j$ is fixed, $j=1,\dots, M$,  $R_\varepsilon$ has the form
\begin{equation}\label{regptransmultpi}
\begin{array}{c}
\displaystyle{R_\varepsilon(\X, \Y)= R^{(\Omega)}(\X, \Y)+h_N^{(j,I)}(\XI_j, \E_j)+\sum_{\substack{k \ne j\\ 1 \le k \le M }}h_N^{(k, O)}(\XI_k, \E_k)}\\ \\
\displaystyle{+(2\pi)^{-1}(\mu_O^{-1}-\mu_{I_j}^{-1})\log \varepsilon-\varepsilon\cD^{(j,I)}(\E_j) \cdot \nabla_\Y R^{(\Omega)}(\X, \Oj^{(j)})}\\\\
\displaystyle{-\varepsilon\sum_{\substack{k \ne j\\ 1 \le k \le M}}\cD^{(k,O)}(\E_k) \cdot \nabla_\Y R^{(\Omega)}(\X, \Oj^{(k)})-\varepsilon\sum^M_{p=1}\cD^{(p)}(\XI_p) \cdot \nabla_\X R^{(\Omega)}(\Oj^{(p)}, \Y) +O(\varepsilon^2)\;.}
\end{array}
\end{equation}

Before proceeding with  examples where these formulae will be implemented, we  first discuss the numerical settings.

\subsection{Numerical settings: Description of the geometry and physical parameters}

Let $\Omega$ be a disk of radius 150m, with centre at the origin, and occupied by a material with shear modulus $\mu_O=5.6\times 10^{10}\text{Nm}^{-2}$, which is that of Cast Iron. We set the number of inclusions  $M=6$, and assume that the $\omega_\varepsilon^{(j)}$, $j=1, \dots, 6$, are circular. We summarize the data corresponding to the inclusions in Table 1.
\begin{table}[htbp]\centering
\label{tabincdata}
\scalebox{0.8}{\begin{tabular}{|c|c|c|c|c|}
 \hline\hline Inclusion & Centre & Radius (m)& Shear Modulus ($\times 10^{10}\text{Nm}^{-2}$) & Material\\
\hline\hline $\omega_\varepsilon^{(1)}$   &(-90m, 40m)&     27      & 2.6316&        Aluminum \\
\hline $\omega_\varepsilon^{(2)}$   &(-50m,-50m)&     24      & 4.0741&        Copper \\
\hline $\omega_\varepsilon^{(3)}$   &(-30m, 10m)&     9      & 7.7519&       Iron \\
\hline $\omega_\varepsilon^{(4)}$   &(20m, 70m)&     19.5      & 7.5188&    High Strength Alloy Steel    \\
\hline $\omega_\varepsilon^{(5)}$   &(50m, 0m)&     22.5      & 8.0078&     Steel AISI 4348 \\
\hline $\omega_\varepsilon^{(6)}$   &(70m, -80m)&     15     & 9.0496&       Nimonic Alloy 90 \\
\hline
\end{tabular}}
\caption{Data for the inclusions $\omega_\varepsilon^{(j)}$, $j=1, \dots, 6$.}
\end{table}

\subsection{Model solutions used in the numerical simulations}

\subsection*{Neumann's function for the disk $\Omega$}

In our examples, we need the Neumann function $N^{(\Omega)}$ for a disk of radius $R$ ($R=150$m for our demonstrations), which is given by 
\begin{equation*}\label{nsant7}
N^{(\Omega)}(\X, \Y )=-\frac{1}{2\pi\mu_O}\log|\X-\Y|-\frac{1}{2\pi\mu_O} \log\left(\left|\X-\frac{R^2}{|\Y|^2}\Y\right||\Y|\right)\;,
\end{equation*} 
and the regular part $R^{(\Omega)}$ of this function is defined by the formula
\begin{equation*}\label{nsant8}
R^{(\Omega)}(\X, \Y)=(2\pi\mu_O)^{-1}\log |\X-\Y|-N^{(\Omega)}(\X, \Y)\;.
\end{equation*}

\subsection*{The regular part of Green's function for the transmission problem in a plane with an inclusion at the origin}

Now, we state the form of the regular part $h_N^{(j)}$ of the Green's function for the transmission problem in the infinite plane with a circular inclusion at the origin of radius $a_j$. The solution is constructed using the involution procedure which is discussed in \cite{HonHer}. The representation of this function, is dependent on the position of the point $\E_j$. When $\E_j \in C \bar{\omega}^{(j)}$
\begin{equation*}\label{nsant9}
h^{(j,O)}(\XI_j, \E_j)=-\frac{1}{2\pi \mu_O} \frac{\mu_{I_j}-\mu_O}{\mu_{I_j}+\mu_O}\log\left(\frac{1}{|\XI_j|} \left| \XI_j-\frac{a_j^2}{|\E_j|^2}\E_j\right|\right)\quad \text{ for }\quad \XI_j \in C\bar{\omega}^{(j)}\;, 
\end{equation*}
and
\begin{equation*}\label{nsant10}
h^{(j,O)}(\XI_j, \E_j)=\frac{1}{2\pi \mu_O} \frac{\mu_{I_j}-\mu_O}{\mu_{I_j}+\mu_O}(\log|\XI_j-\E_j|^{-1}+\log|\E_j|)\;, \quad \XI_j \in \omega^{(j)}\;.
\end{equation*}

For the case $\E_j \in \omega^{(j)}$, $\XI_j \in C\bar{\omega}^{(j)}$
\begin{equation*}\label{nsant11}
h^{(j,I)}(\XI_j, \E_j)=\frac{1}{2\pi } \frac{\mu_{I_j}-\mu_O}{\mu_{I_j}+\mu_O}\left(\frac{1}{\mu_{I_j}}\log|\XI_j-\E_j|+\frac{1}{\mu_O} \log|\XI_j|\right)\;,
\end{equation*}
and for $\E_j, \XI_j \in \omega^{(j)}$ we have
\begin{equation*}\label{nsant12}
h^{(j,I)}(\XI_j, \E_j)=\frac{1}{2\pi} \frac{\mu_{I_j}-\mu_O}{\mu_{I_j}+\mu_O} \left(\frac{1}{\mu_{I_j}}\log\left(\frac{|\E_j|}{a_j}\left|\XI_j-\frac{a_j^2}{|\E_j|^2}\E_j\right|\right)+\frac{1}{\mu_O}\log a_j \right)\;.
\end{equation*}

\subsection*{The dipole fields for the circular inclusion in the infinite plane}

Here, we give the vector function $\cD^{(j)}$ whose components are the dipole fields for the circular inclusion of radius $a_j$ in the infinite plane and
\[ \cD^{(j)}(\XI_j)=\chi_{C\bar{\omega}^{(j)}}(\XI_j)\cD^{(j, O)}(\XI_j)+\chi_{\omega^{(j)}}(\XI_j)\cD^{(j, I)}(\XI_j)\;.\] 

We have
\begin{equation*}\label{nsant13}
\cD^{(j,O)}(\XI_j)=\frac{\mu_{I_j}-\mu_O}{\mu_{I_j}+\mu_O}\frac{a_j^2\, \XI_j}{|\XI_j|^2}\;, \quad \text{ for } \XI_j \in C\bar{\omega}^{(j)}\;,
\end{equation*}
and
\begin{equation*}\label{nsant14}
\cD^{(j,I)}(\XI_j)=\frac{\mu_{I_j}-\mu_O}{\mu_{I_j}+\mu_O} \XI_j\;, \quad \text{ for } \XI_j \in \omega^{(j)}\;.
\end{equation*}

\subsection{Example 1}
\subsubsection*{The case of the force applied outside the inclusions}

For our first example, we look at the case when $\Y \in \Omega_\varepsilon$. We therefore base our computations on the asymptotic formula (\ref{regptransmultpo}) when comparing with those of COMSOL.
The coordinates of the point force are given as $\Y =(-10\text{m}, -80\text{m})$. We plot the modulus of the gradient of the regular part $R_\varepsilon$ in Figure \ref{inclfigpo}a) according to the analytical formulae (\ref{regptransmultpo}). Figure \ref{inclfigpo}b) is the same quantity computed using the method of finite elements in COMSOL. Both figures are very similar, the maximum absolute error between these computations is $7.666 \times 10^{-16}$ occurring on the exterior boundary.

\begin{figure}[htbp]
\begin{minipage}[b]{0.495\linewidth}
\centering
\includegraphics[width=\textwidth]{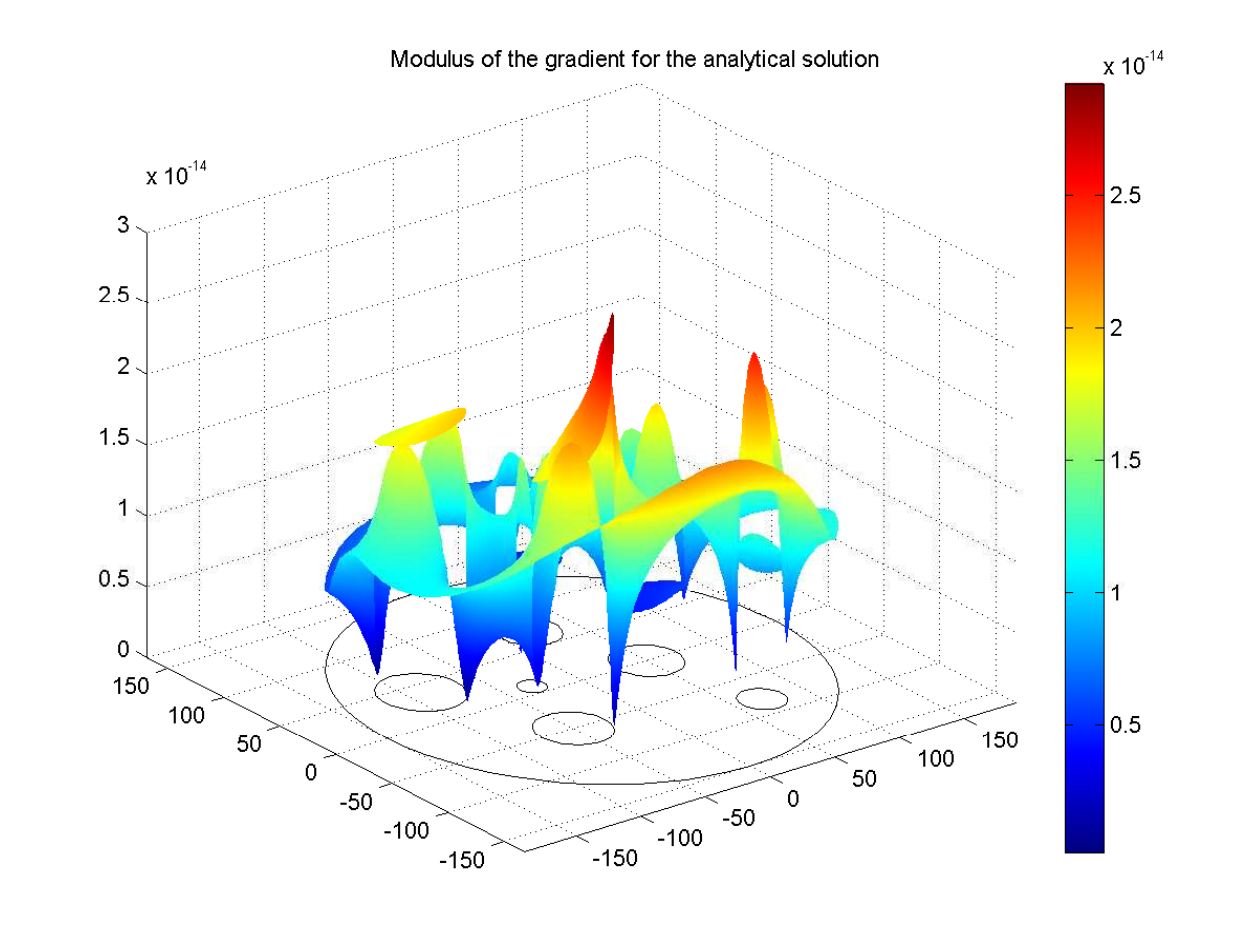}

a)
\end{minipage}
\begin{minipage}[b]{0.495\linewidth}
\centering
\includegraphics[width=\textwidth]{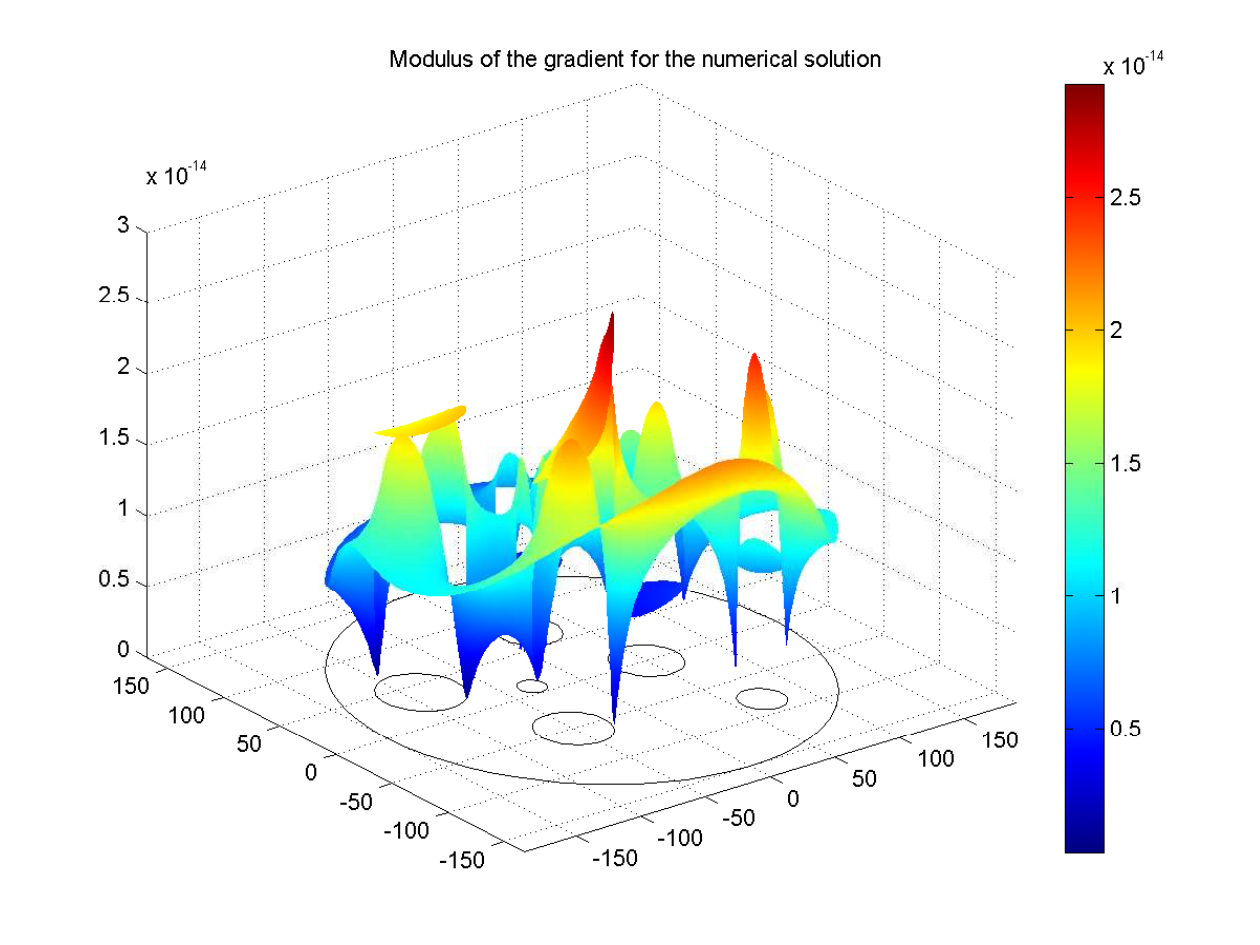}

b)
\end{minipage}
\caption{a) Computations based on  asymptotic formula (\ref{regptransmultpo}) and  b) Numerical solution  for the absolute value of the gradient of the regular part of Green's function for the transmission problem. Here $\Y=(-10\text{m}, -80\text{m})$ and the mesh contains 44784 elements. The plots are practically indistinguishable.}
\label{inclfigpo}
\end{figure}




\subsection{Example 2}
\subsubsection*{The case of the force positioned inside an inclusion}

In the second example, we aim to compare the computations produced by formula (\ref{regptransmultpi}) with those generated by COMSOL. Now the point force is assumed to be situated at $\Y =(60\text{m}, 0\text{m})$, in the inclusion $\omega_\varepsilon^{(5)}$, containing the Steel AISI 4340. Figure \ref{inclfigpi}a), gives the surface plot for the modulus of the gradient of the regular part, provided by formula (\ref{regptransmultpi}). The numerical solution given in COMSOL is shown in Figure \ref{inclfigpi}b). The maximum absolute error here is $7.98 \times 10^{-16}$, which occurs on the boundary of the inclusion $\omega^{(1)}_\varepsilon$.  We conclude that the asymptotic formulae and numerical computations are in a good agreement with each other.

\begin{figure}[htbp]
\begin{minipage}[b]{0.495\linewidth}
\centering
\includegraphics[width=\textwidth]{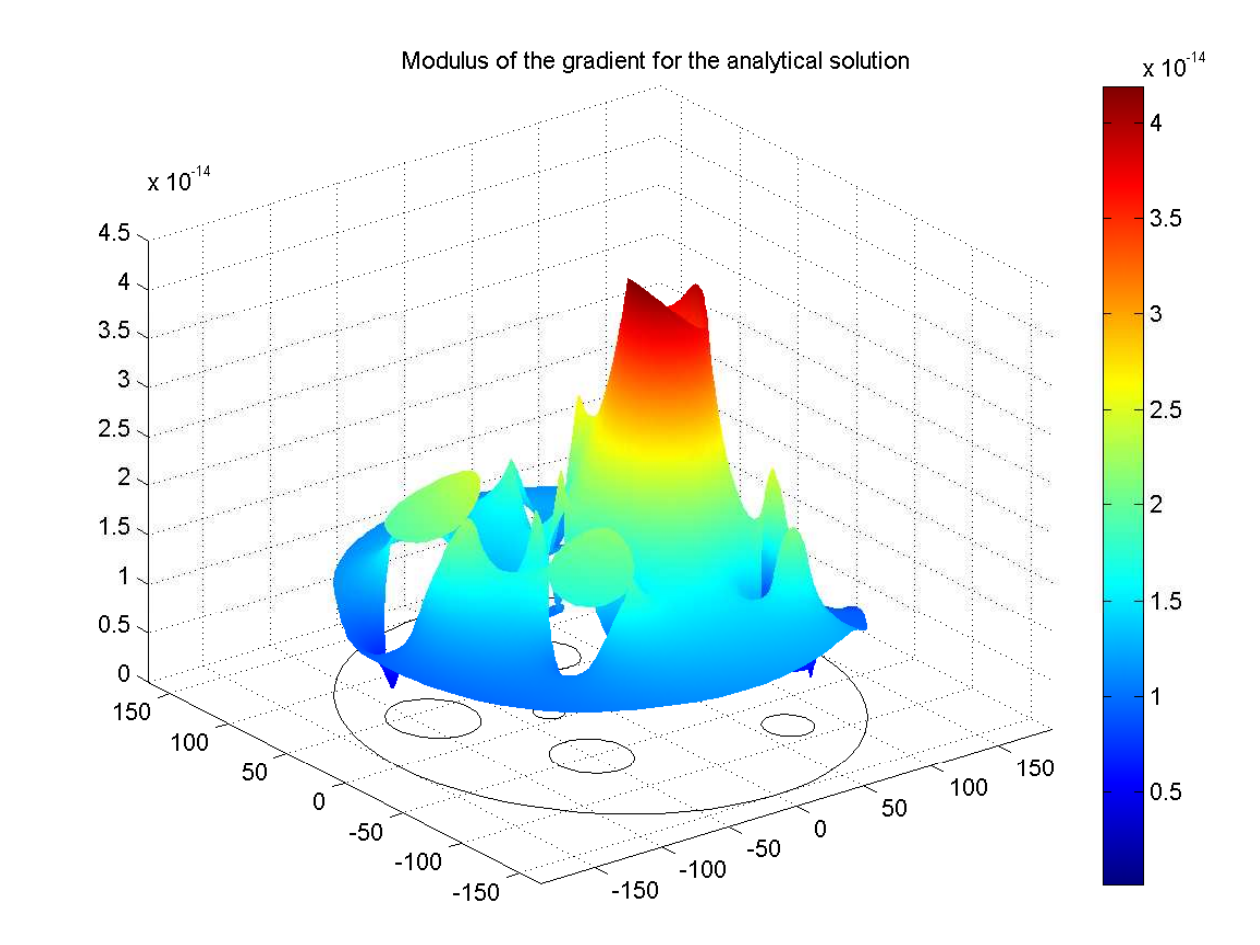}

a)
\end{minipage}
\begin{minipage}[b]{0.495\linewidth}
\centering
\includegraphics[width=\textwidth]{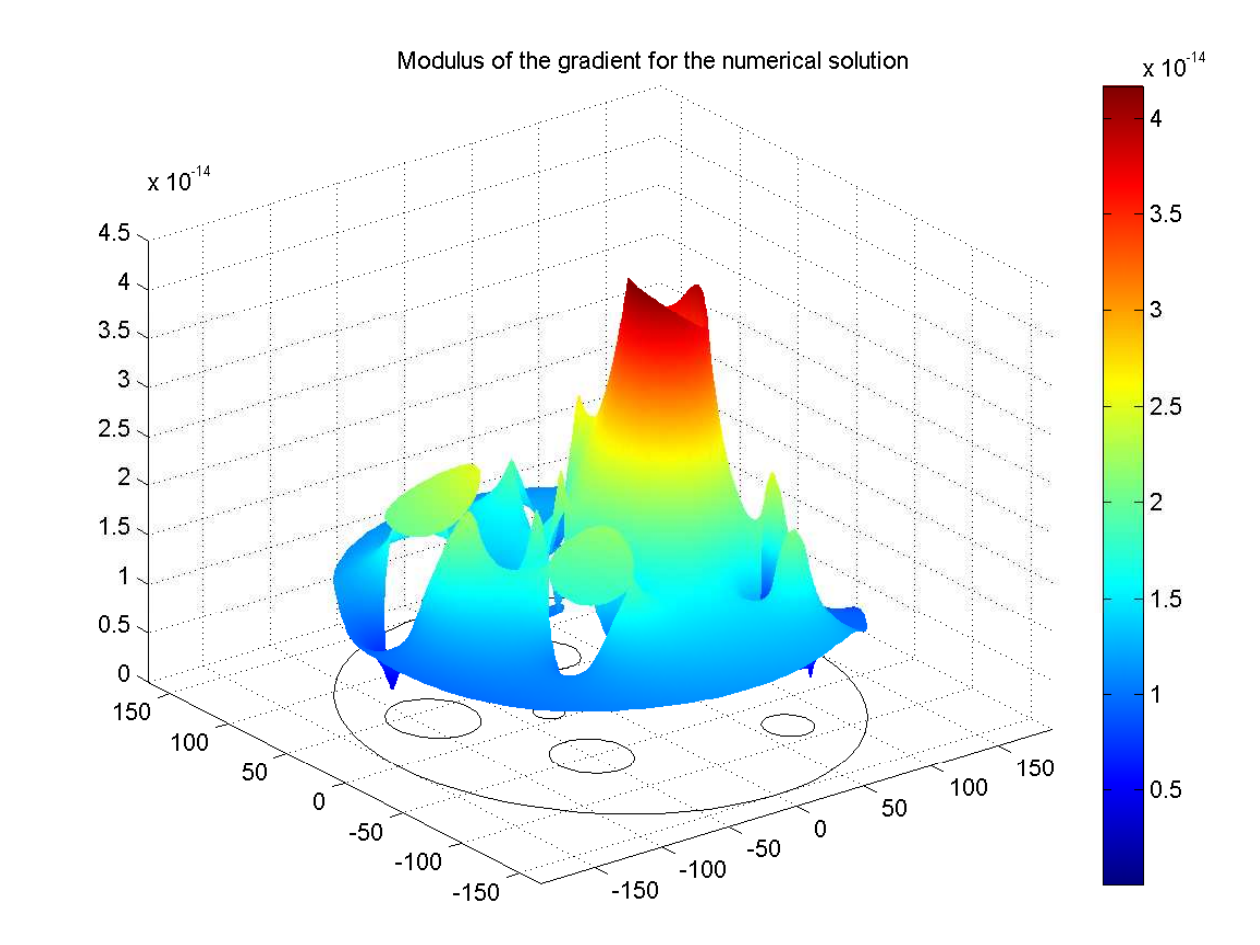}

b)
\end{minipage}
\caption{a) Computations based on  asymptotic formula (\ref{regptransmultpi}) and  b) Numerical solution for the absolute value of the gradient of the regular part of Green's function for the transmission problem. Here $\Y=(60\text{m}, 0\text{m})$ and lies in the Steel AISI 4340 inclusion. The mesh contains 44784 elements. There is once again a good similarity between the surface plots.}
\label{inclfigpi}

\end{figure}

\bibliographystyle{amsalpha}

\end{document}